\def\FullVersion{} %
\def\ShowAuthor{} %
\ifdefined\LLNCS
\documentclass{style/llncs}
\renewcommand{\paragraph}{\subsubsection}
\else
\documentclass[12pt]{article}
\usepackage{amsthm}
\usepackage[margin=1in]{geometry}
\fi

\usepackage{graphicx,color,url,booktabs,comment}
\usepackage{babel}
\usepackage{amsmath, amssymb,amsbsy}
\usepackage{multicol}
\usepackage{framed}
\usepackage{enumitem}
\usepackage{bbm}
\usepackage[x11names]{xcolor}
\usepackage{mathrsfs}
\usepackage[normalem]{ulem}
\usepackage[braket,qm]{qcircuit}
\usepackage{braket}
\usepackage{algorithm}
\usepackage{algpseudocode}
\usepackage{xcolor,float}
\usepackage{pagecolor}
\usepackage[utf8]{inputenc}
\usepackage{tabularx}
\usepackage[pdfstartview=FitH,colorlinks,linkcolor=blue,filecolor=blue,citecolor=blue,urlcolor=blue]{hyperref}
\MakeRobust{\Call}
\usepackage{xspace}
\usepackage{aliascnt}
\usepackage{cleveref} 
\usepackage{mdframed}
\usepackage{amsfonts}
\usepackage{graphicx}
\usepackage{tikz}
\usepackage{tikz-cd}
\usepackage{tikz-qtree}
\usepackage{adjustbox}
\usepackage{caption}
\usepackage{blindtext}
\usetikzlibrary{arrows.meta,calc,decorations.pathmorphing}
\ifdefined\LLNCS
\else 
\usepackage[numbers]{natbib}
\fi

\newcommand{\customnote}[4]{{#4\color{#1}[#2: #3]}}
\newcommand{\setnote}[4]{\ifdefined\Draft\newcommand{#1}[1]{\customnote{#3}{#2}{##1}{#4}}\else\newcommand{#1}[1]{}\fi}

\usepackage{setspace}
\floatstyle{boxed} 
\restylefloat{figure}

\newcommand{\linefill}{\rule{\linewidth}{0.8pt}}

\newcounter{protocol}
\newenvironment{protocol}[1]
{
\refstepcounter{protocol}
\par\setlength{\parindent}{0pt}
\rule{\textwidth}{0.3mm}
\textbf{Protocol~\theprotocol} #1 
\vspace{-2.6mm}

\hrulefill \break
}
{
\hrulefill \break
\par
}

\newcounter{securityGame}

\newcounter{balgorithm}

\newcounter{numConstruction}

\newcounter{simulator}

\ifdefined\LLNCS
\newenvironment{thm}{\begin{theorem}}{\end{theorem}}

\newenvironment{rmk}{\begin{remark}}{\end{remark}}
\newtheorem{cor}[corollary]{Corollary}
\newenvironment{lem}{\begin{lemma}}{\end{lemma}}
\newenvironment{cor}{\begin{corollary}}{\end{corollary}}
\newtheorem{dfn}[definition]{Definition}
\else
\newtheorem{intro}{Theorem}
\newtheorem{introthm}[intro]{Theorem}
\newtheorem{introcor}[intro]{Corollary}
\newtheorem{openproblem}{Open Problem}
\newtheorem{thm}{Theorem}[section]

\newaliascnt{cor}{thm}
\newtheorem{cor}[cor]{Corollary}
\aliascntresetthe{cor}

\newaliascnt{lem}{thm}
\newtheorem{lem}[lem]{Lemma}
\aliascntresetthe{lem}

\newaliascnt{lemma}{thm}
\newtheorem{lemma}[lemma]{Lemma}
\aliascntresetthe{lemma}

\newaliascnt{claim}{thm}
\newtheorem{claim}[claim]{Claim}
\aliascntresetthe{claim}

\theoremstyle{definition}
\newaliascnt{dfn}{thm}
\newtheorem{dfn}[dfn]{Definition}
\aliascntresetthe{dfn}
\fi

\ifdefined\LLNCS
\else
\newtheorem*{thm*}{Theorem}
\newenvironment{theorem}{\begin{thm}}{\end{thm}}
\makeatletter
\newtheorem*{rep@theorem}{\rep@title}
\newcommand{\newreptheorem}[2]{%
\newenvironment{rep#1}[1]{%
 \def\rep@title{#2 \ref{##1}}%
 \begin{rep@theorem}}%
 {\end{rep@theorem}}}
\makeatother

\newreptheorem{thm}{Theorem}
\newreptheorem{lem}{Lemma}
\fi

\crefname{lemma}{Lemma}{Lemmas}
\crefname{claim}{Claim}{Claims}
\crefname{figure}{Figure}{Figures}
\crefname{corollary}{Corollary}{Corollaries}
\crefname{proposition}{Proposition}{Propositions}
\crefname{conjecture}{Conjecture}{Conjectures}
\crefname{definition}{Definition}{Definitions}
\crefname{remark}{Remark}{Remarks}
\crefname{example}{Example}{Examples}
\crefname{algorithm}{Algorithm}{Algorithms}
\crefname{bAlgorithm}{Algorithm}{Algorithms}
\crefname{protocol}{Protocol}{Protocols}
\Crefformat{algorithm}{\textup{Algorithm~#2#1#3}}

\renewcommand{\cref}{\Cref} 

\ifdefined\LLNCS
\newenvironment{proofof}[1]{\begin{proof}[of~#1]}{\end{proof}}
\newenvironment{proofsketch}{\begin{trivlist} \item {\it Proof sketch.}} {\qed\end{trivlist}}
\newenvironment{proofsketchof}[1]{\begin{proofsketch}[of~#1]}{\end{proofsketch}}
\else

\newenvironment{proofsketch}{\begin{trivlist} \item {\it Proof sketch.}} {\qed\end{trivlist}}

\fi

\let\mathbb\relax 
\DeclareMathAlphabet{\mathbb}{U}{msb}{m}{n}

\newcommand{\ie}{{i.e.,\ }}
\newcommand{\eg}{{e.g.,\ }}

\def\cD{\mathcal{D}}
\def\cE{\mathcal{E}}
\def\cF{\mathcal{F}}
\def\cH{\mathcal{H}}

\def\cG{\mathcal{G}}
\def\cI{{\cal I}}

\def\cT{{\cal T}}

\def\bbC{\mathbb{C}}

\def\bbR{\mathbb{R}}

\newcommand{\zo}{\{0,1\}}

\renewcommand{\set}[1]{\left\{#1\right\}}
\newcommand{\norm}[1]{\left\lVert#1\right\rVert}

\DeclareMathOperator*{\E}{\mathbb{E}}

\newcommand{\of}[1]{\left(#1\right)}

\newcommand{\out}{\operatorname{out}}

\newcommand{\Setup}{\mathsf{Setup}}

\newcommand{\Enc}{\mathsf{Enc}}

\newcommand{\Dec}{\mathsf{Dec}}

\newcommand{\Sim}{\mathsf{Sim}}

\newcommand{\anc}{\mathsf{anc}}

\newcommand{\Den}{\cD}
\newcommand{\gray}[1]{\textcolor{Gold4}{#1}}

\newcommand{\onreg}[2]{\ensuremath{{#1}^{\gray{#2}}}}
\newcommand{\ketbra}[1]{\ket{#1}\!\bra{#1}}

\newcommand{\identity}{\mathbbm{1}}

\newcommand{\SWAP}{\ensuremath{\mathsf{SWAP}}}
\def\ot{\otimes}

\newcommand{\PauliGroup}{\mathscr{P}}
\newcommand{\CliffordGroup}{\mathscr{C}}
\DeclareMathOperator{\tr}{tr}

\newcommand{\View}{\mathsf{View}}

\newcommand{\wt}{\widetilde}

\newcommand{\reg}[1]{{\gray{#1}}}

\newcommand{\epr}{\mathsf{epr}}

\newcommand{\aux}{{\mathsf{aux}}}

\newcommand{\QCC}{\mathsf{QCC}}

\newcommand{\epscor}{\varepsilon_{\mathrm{cor}}}
\newcommand{\epspriv}{\varepsilon_{\mathrm{priv}}}
\newcommand{\epsport}{\varepsilon_{\mathrm{port}}}
\title{Quantum Encodings, Private Messages, \\and Communication Complexity}

\setnote{\TODO}{TODO}{red}{}
\setnote{\Enote}{Ercheng}{teal}{\footnotesize}
\setnote{\Mnote}{Miryam}{cyan}{\footnotesize}
\setnote{\Tnote}{Tom}{orange}{\footnotesize}

\ifdefined\ShowAuthor
\author{
    Tom Gur\thanks{University of Cambridge. Supported by ERC Starting Grant 101163189 and UKRI Future Leaders Fellowship MR/X023583/1} \and 
    Miryam Mi-Ying Huang\thanks{Carnegie Mellon University. Supported by Nikolai Mushegian Research Fund and the BNY AI Lab.} \and
    Er-Cheng Tang\thanks{University of Washington. Supported in part by NSF CAREER \#2541085}
}
\else
\author{ }
\date{ }
\fi

\ifdefined \LLNCS \institute{} \fi

\begin{document}
\maketitle

\begin{abstract}
We study the complexity of quantum decomposable randomized encodings (QDRE). We establish a bi-directional connection between the QDRE model and the communication complexity model of quantum private simultaneous message protocols (QPSM), where classical communication is free, and quantum communication is the primary complexity measure. We show the following upper and lower bounds.

\begin{itemize}
    \item For every quantum channel mapping $n$ to $m$ qubits, we give a QPSM protocol with quantum communication $m$, using exponential pre-shared entanglement. In particular, constant-output channels have $O(1)$ quantum communication, and classical-output channels require no quantum communication. 
    \item With $O(n)$ pre-shared entanglement, every one-qubit-output channel has an $O(n)$-quantum-communication QPSM protocol. Conversely, there exists a channel and a constant $c>0$ for which any protocol with at most $cn$ pre-shared entanglement requires $\Omega(n)$ quantum communication.
    \item In contrast, we identify a structured class of quantum channels, Clifford-induced channels, for which there exist QPSM protocols with $O(n)$ pre-shared entangled qubits and only $O(1)$ quantum communication complexity.
\end{itemize}

Through our connection, we obtain analogous upper and lower bounds on the quantum encoding size of QDRE. Our results show that the amount of pre-shared entanglement plays a central role in quantum encoding size and quantum communication complexity.

\end{abstract}

\ifdefined\FullVersion
\fi

\section{Introduction}
Garbled circuits~\cite{yao1986generate}, later abstracted as decomposable randomized encodings~\cite{applebaum2004cryptography,applebaum2006computationally}, provide a way to represent computation by a randomized object that reveals the value of the computation yet hides all other information. These notions have played a central role in cryptography, including efficient secure computation~\cite{yao1986generate,kilian1988founding}, one-time programs and leakage-resilient cryptography~\cite{goldwasser2008one}, identity-based encryption~\cite{dottling2017identity}, and more.  Recent works have formalized and constructed an analogous object, called quantum decomposable randomized encoding (QDRE)\footnote{It was originally named DQRE in \cite{brakerski2022quantum}.} \cite{brakerski2022quantum}, which additionally supports quantum computation.

Informally, a quantum randomized encoding of a quantum circuit $Q$ is another quantum circuit $\hat{Q}$, such that on every input $\rho$, the output $Q(\rho)$ can be recovered from $\hat{Q}(\rho)$ and no other information about $\rho$ is revealed from $\hat{Q}(\rho)$. It is called decomposable if $\hat{Q}(\rho)$ can be executed by first running an input-independent preprocessing that generates setup states $(\sigma_1,\dots,\sigma_n)$, and then each qubit $\rho_i$ of the input state is only processed with the corresponding setup state $\sigma_i$. Decomposability is particularly useful in cryptographic settings where the input state is held across multiple parties.

In existing QDRE constructions~\cite{brakerski2022quantum,bartusek2021round,bartusek2024quantum}, the number of qubits in the encoding $\hat{Q}(\rho)$, which we will refer to as the \emph{quantum encoding size}, grows linearly with the circuit size $|Q|$ of the original quantum circuit. However, as the number of qubits is a much scarcer resource than the number of gates in a quantum circuit, this raises a fundamental complexity question: is the linear dependence on the circuit size inherent, or can the quantum encoding size be characterized more tightly?

\subsection{Our results}
We begin with a conceptual contribution that will allow us to study the aforementioned question through the lens of communication complexity.

\subsubsection{Quantum encodings via communication complexity}
We establish a bi-directional connection between QDREs and quantum private simultaneous message protocols (QPSM) for inherently quantum problems, which preserves the relevant complexity measures.

The QPSM model is a privacy-preserving strengthening of the communication complexity model of quantum simultaneous message protocols (QSMP) for inherently quantum problems. A $k$-party QSMP for a quantum computation $Q$ consists of $k$ parties with pre-shared entangled states\footnote{This generalizes the classical setting of private simultaneous messages (PSM), where the parties have unlimited pre-shared classical randomness by default.} and local inputs $\rho_i$, where each party sends the referee a quantum message from which the referee computes the desired output $Q(\rho_1,\dots,\rho_k)$. A QPSM protocol strengthens a QSMP protocol by additionally requiring that the joint quantum messages reveal nothing about the inputs beyond what can be inferred from the output.

Our first result establishes the following connection between QDRE and QPSM, showing an equivalence between QDREs for $k$-qubit inputs decomposed qubit-by-qubit and $k$-party QPSM protocols with one-qubit local inputs, preserving quantum encoding size and quantum communication.

\begin{introthm}[Informally stated; see \cref{thm:QDRE-QPSM-equivalence}]
\label{thm:informal:QDRE-QPSM}
For any quantum circuit $Q$ with $k$-qubit inputs, QDREs for $Q$ with $\alpha$-correctness and $\beta$-privacy are equivalent to $k$-party QPSM protocols for $Q$ with $\alpha$-correctness and $\beta$-privacy where each party receives a $1$-qubit input. Moreover, the correspondence preserves the total encoding size/communication complexity, and also preserves the setup-state size/pre-shared-entanglement complexity.
\end{introthm}

Building on the connection in \cref{thm:informal:QDRE-QPSM}, we proceed to prove the following theorems regarding the complexity of quantum private simultaneous message protocols, which provide the technical core of this paper.
In the following, we focus on the \emph{quantum communication complexity} of QPSM protocols as the primary computational resource; that is, we count the number of qubits that the parties communicate, while allowing free classical communication.

\subsubsection{Quantum private simultaneous message protocols}
We show that, perhaps surprisingly, every quantum channel $\Phi$ that takes $k$ inputs and outputs $m_q$ qubits and $m_c$ classical bits admits a $k$-party QPSM protocol whose quantum communication complexity is equal to the number of qubits $m_q$ that the referee outputs, provided that the parties share a large number of entangled qubits.
As a special case, every classical-output quantum channel (but possibly with quantum inputs) has a QPSM protocol that \emph{does not require any quantum communication}. In the following $\Den(\mathcal{H}; S)$ denote the set of density operators describing classical-quantum states with a quantum register with Hilbert space $\mathcal{H}$ and a classical register taking values in set $S$. We abbreviate $\Den(\mathcal{H}) = \Den(\mathcal{H}; \emptyset)$.

\begin{introthm}[Informally stated; see \cref{thm:upperbound}]
\label{thm:informal:upper-bound}
Let $ \Phi:\Den(\cH_1\otimes \cdots \otimes \cH_k) \to  \Den((\bbC^2)^{\otimes m_q}; \{0,1\}^{m_c})$ be a $k$-party quantum channel. Then $\Phi$ admits a QPSM protocol with quantum communication complexity at most $m_q$ with an arbitrarily small error.
\end{introthm}

The theorem above provides a strong upper bound through a QPSM protocol whose quantum communication complexity is independent of both the complexity of the channel and the input length; in particular, for channels with constant-size outputs, it yields $O(1)$ quantum communication complexity. The drawback of the protocol is that it strongly relies on the entanglement shared between the $k$-parties.

Indeed, the number of pre-shared entangled qubits required in \cref{thm:informal:upper-bound} is exponential in the input length. This motivates a model of quantum private simultaneous message protocols where the pre-shared entangled state is also charged as a complexity measure.

\subsubsection{Bounded-entanglement QPSM}
In the following, we will refer to the number of pre-shared entangled qubits of a QPSM protocol as the \emph{entanglement complexity}. The next theorem studies QPSM protocols with linear entanglement complexity. It provides such protocols with linear quantum communication complexity for every quantum channel that outputs a single qubit. The theorem also shows a channel for which linear quantum communication is unavoidable for QPSM protocols with sublinear entanglement complexity . Moreover, for the family of ``Clifford-induced channels'' (obtained by applying a Clifford unitary to the input together with an ancilla state initialized to $\ket{0}$ and then tracing out some registers), there exist QPSM protocols with linear entanglement complexity, using only $O(1)$ quantum communication.
In the following, we will use $n$ to denote the input size.

\begin{introthm}[Informally stated; see \cref{thm:ublb}]
\label{thm:informal:lower-bound}
The following bounds hold for bounded-entanglement QPSM protocols.
\begin{enumerate}
\item \textbf{Upper bound for general channels.} Every quantum channel $\Phi:\Den(\cH_1\ot\cdots\ot\cH_k)\to \Den(\bbC^2)$ admits a QPSM protocol with $O(n)$ entanglement complexity  and $O(n)$ quantum communication complexity.
\item \textbf{Upper bound for Clifford-induced channels.} Every Clifford-induced quantum channel $\Phi: \Den(\cH_1\ot \cdots \ot \cH_k) \to \Den(\bbC^2)$ admits a QPSM protocol with $O(n)$ entanglement complexity  and $1$-qubit quantum communication complexity.
\item \textbf{Lower bound.} There exists a two-party quantum channel $\Phi: \Den((\bbC^2)^{\otimes n} \otimes (\bbC^2)^{\otimes n}) \to \Den(\bbC^2)$ and constants $c,c'>0$, such that any QPSM protocol computing $\Phi$ with at most $c \cdot n$ entanglement complexity  and correctness error within $c'$ must have quantum communication complexity $\Omega(n)$.
\end{enumerate}
\end{introthm}

In fact, the lower bound in \cref{thm:informal:lower-bound} holds even without the privacy requirement,
namely for QSMP protocols, showing that non-trivial general-purpose QSMP protocols with sublinear quantum communication and entanglement complexity cannot exist. In contrast, we also show that privacy can strictly increase quantum communication: there is a state-generation task with zero QSMP quantum communication but requires nonzero QPSM quantum communication (see \cref{thm:bb84}).

We now return to the motivating question on quantum encoding size. Since the QDRE-QPSM correspondence in \cref{thm:informal:QDRE-QPSM} translates quantum encoding size into quantum communication complexity, and setup-state size into entanglement complexity, the QPSM bounds above immediately imply the following corollaries for QDREs.

\begin{introcor}
\label{cor:correspondence}
The following bounds hold for QDREs.
\begin{enumerate}
\item \textbf{Large setup states.}
Every quantum channel $\Phi:\Den((\bbC^2)^{\ot n})
\to
\Den((\bbC^2)^{\otimes m_q}; \{0,1\}^{m_c})$ admits a QDRE whose quantum encoding size is at most $m_q$ with arbitrarily small error. In particular, every classical-output quantum channel admits a QDRE with zero quantum encoding size.
\item \textbf{Bounded setup states.} Every quantum channel
$\Phi:\Den((\bbC^2)^{\ot n})\to\Den(\bbC^2)$ on $n$ input qubits admits a QDRE with $O(n)$ setup-state size and $O(n)$ quantum encoding size. Moreover, there exists a quantum channel
$\Phi:\Den((\bbC^2)^{\otimes n})
    \to
    \Den(\bbC^2)$
and constants $c,c'>0$ such that any QDRE for $\Phi$ with setup-state size at most $c\cdot n$ and correctness error within $c'$ must have quantum encoding size $\Omega(n)$.

\item \textbf{Clifford-induced quantum channels.} Every Clifford-induced quantum channel
$\Phi:\Den((\bbC^2)^{\ot n})\to\Den(\bbC^2)$ admits a QDRE with $O(n)$ setup-state size and $1$-qubit quantum encoding size.
\end{enumerate}
\end{introcor}

These corollaries show that the circuit size $|Q|$ is not the right sole predictor of quantum encoding size. With sufficiently large setup states, the quantum part of the encoding can depend only on the number of output qubits; under a small linear setup-state size, however, there is a worst-case $\Omega(n)$ barrier in the input length.

\subsection{Techniques}
In this technical overview, we focus on the main theme of this work: how to overcome the challenges that arise from the \emph{privacy} requirement of the QPSM model. For simplicity, we consider an arbitrary quantum channel
$$\Phi : \mathcal{D}(\mathcal{H}_A\otimes\mathcal{H}_B)\rightarrow \mathcal{D}(\mathbb{C}^2)$$
that takes two $n$-qubit inputs $\rho_A,\rho_B$ and produces a
single-qubit output $\Phi(\rho_A,\rho_B)$.  We show bounds on the amount of quantum
communication that is needed to privately compute such a channel when the
inputs are distributed between parties, Alice and Bob, who simultaneously send quantum messages to a referee, Charlie, who outputs the outcome of the computation.

\paragraph{Privacy is a communication constraint: QSMP-vs-QPSM separation.}

In a regular quantum simultaneous-message protocol (QSMP), Charlie only needs to
recover the prescribed output.  In the QPSM model,
Charlie's view must additionally be simulatable from the output
alone.  In particular, even though classical communication is free in
our complexity measure, \emph{classical messages are still
subject to the privacy requirement}.

This requirement can change the quantum communication
complexity.  We illustrate this with a simple distributed BB84-state
generation task here. Alice receives a classical bit $x$, Bob receives a
classical bit $\theta$, and Charlie obtains $H^\theta |x\rangle.$ Without privacy, Alice and Bob can simply send $x$ and $\theta$ to
Charlie, who prepares the state locally.  Thus, a perfectly correct
QSMP protocol requires zero quantum communication.  In contrast, we
show in \cref{thm:bb84} that every QPSM protocol must communicate at least one qubit, even
when Alice and Bob may share arbitrary entanglement.

The intuition is as follows. Suppose there is a QPSM protocol for this task
that only uses classical communication.  The privacy condition requires its classical
transcript to be simulatable from a single copy of the output state,
while correctness implies that Charlie could use this transcript to
approximately reconstruct that same output.  Composing these two
procedures would give a measure-and-prepare channel that approximately
preserves all four BB84 states.  Since these states are non-orthogonal,
no such classical intermediate representation can preserve all of them
with sufficiently small error.

Thus, privacy can force quantum communication even when all of the
parties' inputs are classical.  This raises a natural question:
for an arbitrary one-qubit-output channel, how much quantum
communication is sufficient?

Perhaps surprisingly, one qubit always suffices, provided that the parties
share sufficiently large pre-shared entanglement.  More generally,
\cref{thm:informal:upper-bound} shows that for an arbitrary quantum
channel, the quantum communication can be reduced all the way to the
number of quantum output qubits.
To see why this is possible, it is useful to first consider two natural
approaches based on teleportation.

\paragraph{First attempt: ordinary teleportation.}
To see how such an output-size upper bound can be achieved, a
natural first attempt is to let Alice teleport her input $\rho_A$ to Bob.
Consequently, Alice will hold a random $2n$-bit classical string describing
a Pauli correction $P$, while Bob will hold the $2n$-qubit quantum state
$(P\rho_A P^\dagger,\rho_B)$.

If Bob knew $P$, he could perform the correction and apply $\Phi$ to obtain the
final output by himself and send only the one-qubit output to Charlie.
However, in the simultaneous-message setting, Alice can only send $P$ to
Charlie, not to Bob.  To reduce quantum communication, Bob would therefore
have to compress the Pauli-masked state
$(P\rho_A P^\dagger,\rho_B)$ without knowing the Pauli correction, while
leaving enough quantum information for Charlie, who can learn $P$ from
Alice, to recover the true output.  Achieving sublinear compression of this
kind for a general quantum channel $\Phi$ seems challenging.

\paragraph{Second attempt: port-based teleportation.}
To overcome the aforementioned challenge, one could try to replace quantum
teleportation with a powerful type of teleportation protocol called
port-based teleportation
\cite{ishizaka2008asymptotic,ishizaka2009quantum},
in the hope of avoiding the need for compressing Pauli-masked states.
In port-based teleportation, the teleported state does not require
correction\footnote{The teleported state will be an approximation of the
original state.}. Instead, the teleported state appears in one of many
ports held by the receiver, and the sender learns the classical label of
the port holding the teleported state. Thus, port-based teleportation
trades an unknown Pauli correction on the data for an unknown port index.

Consider a protocol where Alice performs port-based teleportation of her
input, so that Bob receives Alice's input in one of his ports. If Bob knew
the correct port, he could then compute $\Phi$ on their joint input and
send the single-qubit output to Charlie. However, Alice cannot communicate
the correct port index to Bob, and it is unclear what Bob can do to reduce
quantum communication without knowing the correct port index. In
particular, Bob cannot simply apply $\Phi$ independently to every candidate
port, since he has only one copy of his own unknown quantum input.

\paragraph{Combining the two: Pauli-Port Complementarity.}
The two approaches above have complementary difficulties.
Ordinary teleportation yields a Pauli-corrupted state, but the Pauli
correction is known to the other party. Port-based teleportation removes
the need for a Pauli correction, but hides the location of the correct
output port. Combining these two ideas allows us to use one type of
uncertainty to handle the other: the party who knows the Pauli correction
can apply the appropriately corrected computation to every port, while
the party who knows the correct port later selects the corresponding
output. The resulting protocol is illustrated in
\cref{fig:qpsm-pbt-overview}. First, Bob teleports $\rho_B$ to Alice using
ordinary teleportation. Writing $\rho_{AB}$ for their joint input state,
Alice obtains
$$\widetilde{\rho}_{AB}=(I\otimes P)\rho_{AB}(I\otimes P^\dagger),$$
and Bob knows the Pauli correction $P$. 
Alice then uses port-based teleportation to teleport this \emph{entire
joint state} to Bob. Bob holds ports
$g_{B,1},\ldots,g_{B,N}$, while Alice learns a port label $i\in[N]$ such
that the $i$-th port contains a state close to
$\widetilde{\rho}_{AB}$. Each port now contains registers for both inputs,
so Bob can process every port without needing additional copies of either
input.

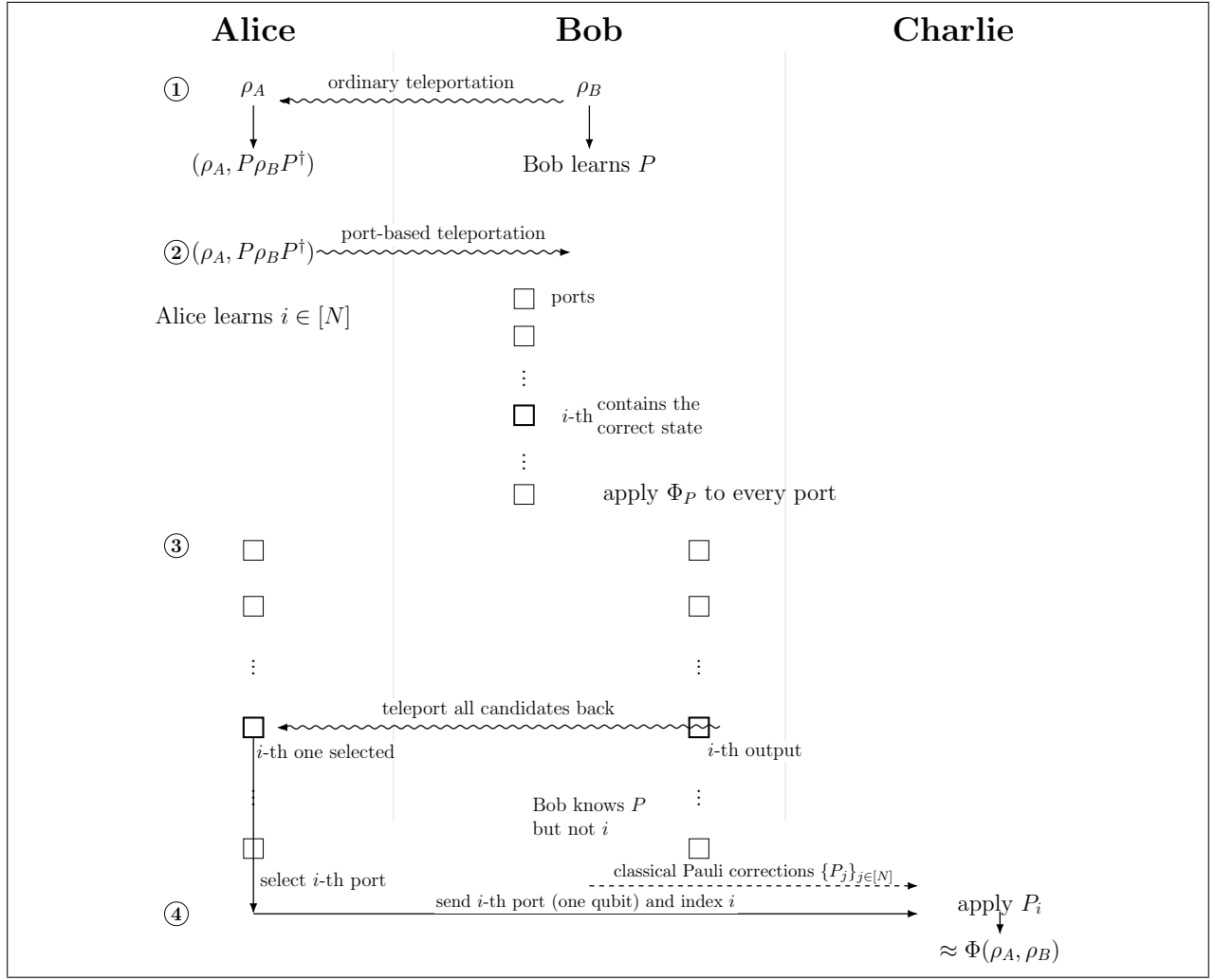
\begin{figure*}[ht!]
\centering
\begin{adjustbox}{scale=0.65,center}
\begin{tikzpicture}[
    x=1cm,y=1cm,
    >=Latex,
    font=\large,
    qarrow/.style={->, line width=0.75pt},
    carr/.style={->, dashed, line width=0.75pt},
    tele/.style={
        ->,
        line width=0.75pt,
        decorate,
        decoration={snake, amplitude=1pt, segment length=8pt}
    },
    port/.style={
        draw,
        minimum width=0.42cm,
        minimum height=0.42cm,
        inner sep=0pt
    },
    goodport/.style={
        draw,
        very thick,
        minimum width=0.42cm,
        minimum height=0.42cm,
        inner sep=0pt
    },
    lab/.style={
        font=\normalsize,
        fill=white,
        inner sep=1.2pt
    },
    smalllab/.style={
        font=\small,
        fill=white,
        inner sep=1pt
    },
    partytitle/.style={
        font=\bfseries\LARGE
    },
    stepnode/.style={
        circle,
        draw,
        inner sep=1.5pt,
        font=\normalsize\bfseries
    }
]

\node[partytitle] at (2.0,16.0) {Alice};
\node[partytitle] at (9.2,16.0) {Bob};
\node[partytitle] at (17.0,16.0) {Charlie};

\draw[gray!25] (5.0,-1.0) -- (5.0,15.5);
\draw[gray!25] (13.4,-1.0) -- (13.4,15.5);

\node[stepnode] at (0.35,14.7) {1};

\node at (2.0,14.7) {$\rho_A$};
\node at (9.2,14.7) {$\rho_B$};

\draw[tele]
    (8.65,14.45)
    --
    node[midway, above=4pt, lab]
    {ordinary teleportation}
    (2.55,14.45);

\node at (2.0,13.1) {$(\rho_A, P\rho_B P^\dagger)$};
\node at (9.2,13.1) {Bob learns $P$};

\draw[qarrow] (2.0,14.35) -- (2.0,13.4);
\draw[qarrow] (9.2,14.35) -- (9.2,13.4);

\node[stepnode] at (0.35,11.2) {2};

\node at (2.0,11.2) {$(\rho_A, P\rho_B P^\dagger)$};

\draw[tele]
    (3.35,11.2)
    --
    node[midway, above=4pt, lab]
    {port-based teleportation}
    (8.8,11.2);

\node at (2.0,9.8) {Alice learns $i\in[N]$};

\node[port]     (bp1) at (7.8,10.2) {};
\node[port]     (bp2) at (7.8,9.4) {};
\node            (bpd1) at (7.8,8.6) {$\vdots$};
\node[goodport] (bpi) at (7.8,7.7) {};
\node            (bpd2) at (7.8,6.8) {$\vdots$};
\node[port]     (bpN) at (7.8,6.0) {};

\node[lab] at (8.85,10.2) {ports};
\node[lab] at (8.95,7.7) {$i$-th};
\node[lab, align=left] at (10.5,7.7)
    {contains the\\correct state};

\node[stepnode] at (0.35,4.9) {3};

\node at (12,6) {apply $\Phi_P$ to every port};

\node[port]     (cp1) at (11.55,4.8) {};
\node[port]     (cp2) at (11.55,3.6) {};
\node            (cpd1) at (11.55,2.4) {$\vdots$};
\node[goodport] (cpi) at (11.55,1.0) {};
\node            (cpd2) at (11.55,-0.4) {$\vdots$};
\node[port]     (cpN) at (11.55,-1.6) {};

\node[lab] at (12.75,0.5) {$i$-th output};

\node[port]     (ap1) at (2.0,4.8) {};
\node[port]     (ap2) at (2.0,3.6) {};
\node            (apd1) at (2.0,2.4) {$\vdots$};
\node[goodport] (api) at (2.0,1.0) {};
\node            (apd2) at (2.0,-0.4) {$\vdots$};
\node[port]     (apN) at (2.0,-1.6) {};

\node[lab] at (3.55,0.5) {$i$-th one selected};

\draw[tele]
    (12,1.02)
    --
    node[midway, above=4pt, lab]
    {teleport all candidates back}
    (2.50,1.0);

\node[lab, align=left] at (9.2,-0.9)
    {Bob knows $P$\\but not $i$};

\node[stepnode] at (0.35,-3.0) {4};

\draw[qarrow] (api.south) -- (2.0,-3.0);

\node[lab] at (3.5,-2.3) {select $i$-th port};

\draw[qarrow]
    (2.0,-3.0)
    --
    node[above, lab, font=\small]
    {send $i$-th port (one qubit) and index $i$}
    (16.25,-3.0);

\draw[carr]
    (9.2,-2.4)
    --
    node[above, lab, font=\small]
    {classical Pauli corrections $\{P_j\}_{j\in[N]}$}
    (16.25,-2.4);

\node at (18,-2.8) {apply $P_i$};
\node at (18,-3.8) {$\approx \Phi(\rho_A,\rho_B)$};

\draw[qarrow] (18.0,-2.95) -- (18.0,-3.45);

\end{tikzpicture}
\end{adjustbox}

\caption{Pauli-Port Complementarity Protocol. The protocol is non-interactive. Teleportation arrows denote local teleportation measurements using pre-shared entanglement, not communication between Alice and Bob.}
\label{fig:qpsm-pbt-overview}
\end{figure*}

Note that Bob does not know which port is correct, but he knows the Pauli correction
$P$. Therefore, for every port $j\in[N]$, Bob can apply the corrected channel
$$\Phi_P(X):=\Phi\bigl((I\otimes P^\dagger)X(I\otimes P)\bigr).$$
This produces candidate one-qubit outputs
$$\rho'_1,\ldots,\rho'_N.$$
The $i$-th candidate is close to the desired output
$\Phi(\rho_A,\rho_B)$, but Bob does not need to know which one is correct.
He teleports all candidate outputs back. Alice, who knows $i$, sends only
the $i$-th single-qubit register and the index $i$ to Charlie, while Bob
sends all the final Pauli correction information. Charlie applies the
$i$-th Pauli correction and obtains the final output.
The protocol remains non-interactive: Alice and Bob never
communicate with each other during its execution.  Each teleportation
step consists only of local operations and measurements on pre-shared
entanglement; the only communication occurs at the end, when Alice and
Bob simultaneously send their respective messages to Charlie.
Despite the large number of
intermediate ports, the quantum communication to the referee is only one
qubit because all candidate outputs are teleported to Alice, who knows how to select the correct one. The same selection idea extends to the $k$-party setting: the
parties perform the corresponding port-selection steps locally and
simultaneously, while only the final selected output register is sent by
the designated party to the referee.  Consequently, for a one-qubit-output
channel, the resulting $k$-party protocol still uses only one qubit of
quantum communication.

Correctness follows from the correctness of port-based teleportation and
monotonicity of trace distance under quantum channels. The correct port
contains a state close to the Pauli-masked joint input, and after Bob
applies $\Phi_P$ to every port, the selected candidate output is close to
$\Phi(\rho_A,\rho_B)$.

We remark that a closely related teleportation architecture appears in the
instantaneous non-local computation protocol of Beigi and
K\"onig \cite{beigi2011simplified}, which also combines ordinary teleportation and
port-based teleportation to implement a non-local quantum operation. They consider communication protocols that are weaker than private simultaneous message protocols.
Moreover, while \cite{beigi2011simplified} studies whether non-local computation can be implemented, our work targets privacy-preserving protocols, and treats the
amount of quantum information communicated to the referee as a central
complexity measure. In particular, directly implementing a non-local computation is not sufficient for our purposes if doing so requires communicating many candidate quantum registers. 
Additionally, correctness is only part of the requirement: to obtain
a QPSM protocol, Charlie's \emph{entire view} must additionally be
simulatable from the prescribed output alone.

As illustrated by
our QSMP-QPSM separation above, information that is sufficient for
correctness may reveal more than the prescribed output, and requiring
privacy can itself change the communication complexity of the problem.
This distinction is also familiar from the classical PSM literature,
where privacy leads to a substantially different communication complexity landscape than standard simultaneous-message computation
\cite{feige1994minimal,beimel2014cryptographic,applebaum2020communication}.

In our setting, the difficulty is particularly visible because Charlie
receives both classical information and a quantum register, and privacy
must hold for their \emph{joint} state.

\paragraph{From correctness to privacy.}
The main challenge is to achieve correctness while ensuring privacy of Charlie’s entire joint view. Correctness requires Charlie to recover
$\Phi(\rho_A,\rho_B)$, while privacy requires that Charlie's entire view must be reproducible from this output alone.
The permutation symmetry of the port-based teleportation measurement provides a useful starting point: no port is distinguished a priori, and hence every port is equally likely to be selected. However, this marginal uniformity alone is not sufficient for privacy, because the selected index may still be correlated with the quantum state contained in the selected port.
Our privacy proof therefore analyzes these quantities jointly.  The key
step is a strengthened view of port-based teleportation: instead of
discarding the port label and considering only the selected port, we keep
both in the state. \cref{lemma:extended-correctness-pbt} shows that the selected index $i$ is uniform
and independent of the selected register, even in the presence of
arbitrary quantum side information, while the selected register remains
close to the desired teleported state.
Another important feature of the protocol is that the unselected ports
never appear in Charlie's view due to Alice's selection of the correct output candidate. Thus, the privacy analysis only needs to characterize
the selected register together with the classical information that Charlie
actually receives.
We then carry this joint-state
guarantee through the computation and the final teleportations.

More precisely, if $\mathsf{View}$ denotes everything received by Charlie,
our analysis constructs a simulator $\mathsf{Sim}$ such that, for every
input possibly entangled with an external reference register $R$, $$  \frac12
  \left\|
    (\mathsf{View}\otimes I_R)(\rho_{ABR})
    -
    (\mathsf{Sim}\otimes I_R)
    \bigl((\Phi\otimes I_R)(\rho_{ABR})\bigr)
  \right\|_1
  \leq \varepsilon .$$

The simulator samples a fresh port index and fresh Pauli corrections,
and applies the corresponding Pauli mask to the ideal output.  Thus,
although the protocol generates a large amount of intermediate
teleportation data, its entire joint view can be
reconstructed from the prescribed output alone.

Thus, for one-qubit-output functionalities, there is a QPSM protocol with $$\QCC(\Phi)\le 1.$$ More generally, if the output contains $m_q$ qubits, Alice sends only the selected $m_q$-qubit candidate output, giving $\QCC(\Phi)\le m_q$. If the output is purely classical, Alice can send the selected classical outcome instead, and the quantum communication becomes zero.

\subsection{Related work}
Randomized encodings and decomposable randomized encodings (DRE) have been extensively studied in classical cryptography, both as abstractions of garbled circuits and as tools for secure computation.  Lower bounds are known for several relaxed notions of randomized encodings \cite{applebaum2014key,agrawal2015statistical}, but lower bounds for DRE remain subtle.

A closely related classical model is private simultaneous messages (PSM), introduced by Feige, Kilian, and Naor \cite{feige1994minimal} and further studied by Ishai and Kushilevitz \cite{ishai1997private}.  In a PSM protocol, each party holds one input share and sends one message to a referee, who should recover only the target value.  Classical DRE can be viewed as PSM protocols in which each party holds one input bit.  Thus, PSM lower bounds imply lower bounds on decomposable encoding size, while PSM upper bounds yield decomposable encodings under an appropriate input decomposition. Recent works have made this connection quantitative by proving communication lower bounds for PSM and decomposable encodings \cite{applebaum2020communication,ball2020complexity,ball2022note}. 

Quantum decomposable randomized encodings were studied by Brakerski and Yuen \cite{brakerski2022quantum}, who provided a construction consisting of local encodings of the input qubits together with a circuit-dependent offline component. They additionally considered a universal QDRE, whose privacy guarantee hides both the input and the circuit. We consider the ordinary QDRE notion, where the quantum circuit is treated as fixed and public.

Quantum variants of PSM have also been studied, mostly for classical inputs and classical functionalities.  Kawachi and Nishimura introduced private simultaneous quantum messages, where parties hold classical inputs but may send quantum messages and use shared entanglement \cite{kawachi2021communication}.  More recent work studies quantum-assisted PSM protocols for classical functionalities, obtaining new upper and lower bounds using quantum communication, shared entanglement, and measures such as rank and Nečiporuk-type complexity \cite{cryptoeprint:2026/1236}.  Other related directions connect PSM-style primitives to non-local quantum computation, position verification, and coherent function evaluation \cite{allerstorfer2024relating}.  These works show that quantum resources can significantly affect PSM complexity.  In contrast, our QPSM model is fully quantum: the parties may receive quantum input registers, the functionality is a quantum channel, and the referee may output a quantum state.

\subsection{Open problems}
Our work leaves many interesting open problems. We highlight two directions below.

Recall that \cref{thm:informal:upper-bound} shows a general-purpose QPSM protocol which is extremely efficient in terms of quantum communication by using an exponential amount of pre-shared entanglement. While \cref{thm:informal:lower-bound} gives a worst-case linear lower bound when the parties share at most a linear number of entangled qubits, it does not rule out the possibility that the entanglement complexity for the general-purpose upper bound in \cref{thm:informal:upper-bound} could be improved from exponential to polynomial.

\begin{openproblem}
Does every quantum channel $\Phi:\Den(\cH_A\otimes\cH_B) \to  \Den((\bbC^2)^{\otimes m_q}; \{0,1\}^{m_c})$ admit a QPSM protocol with polynomial entanglement complexity  and quantum communication complexity at most $m_q$?
\end{openproblem}

A second direction is to better understand the cost of privacy itself. In the classical setting with single-bit outputs, the best-known general upper bounds for private simultaneous message protocols are exponentially larger than the trivial upper bound for non-private simultaneous message
protocols \cite{feige1994minimal,applebaum2020communication,ball2022note}. Our QPSM vs QSMP separation shows that there is a single-qubit-output quantum channel for which privacy requires quantum communication, whereas classical communication alone suffices without the privacy condition. It is natural to ask whether the privacy requirement can cause an asymptotic blowup in quantum communication complexity when computing single-qubit-output channels with bounded pre-shared entanglement.
                                                                                                                            
\begin{openproblem}
For an entanglement bound $e(n)$, does there exist a quantum channel $\Phi: \cD((\bbC^2)^{\ot n}) \to \cD(\bbC^2)$ for which QSMP protocols have $O(1)$ quantum complexity, while QPSM protocols with the same entanglement bounds require $\omega(1)$ quantum complexity?
\end{openproblem}
\section{Preliminaries} \label{sec:prelim}


\paragraph{Quantum states and registers.}
An $n$-qubit pure quantum state $\ket{\phi}$ is a unit vector in the Hilbert space $(\bbC^{2})^{\ot n}$, and is identified with the density matrix $\phi = \ketbra{\phi} \in \bbC^{2^n \times 2^n}$. The space of all probability distributions over an $m_q$-qubit pure state together with an $m_c$-bit string is written as $\Den((\bbC^2)^{\otimes m_q}; \{0,1\}^{m_c})$. We sometimes simplify the notation as $\Den((\bbC^2)^{\otimes m_q})$ (resp. $\Den(\{0,1\}^{m_c})$) when $m_c = 0$ (resp. $m_q = 0$). In particular, $\Den((\bbC^{2})^{\ot n})$ is equal to the set of $n$-qubit mixed states. A register $\reg{R} = (\reg{R^{\mathrm c}}, \reg{R^{\mathrm q}})$ consists of a quantum register $\reg{R^{\mathrm q}}$ that stores a $|\reg{R^{\mathrm q}}|$-qubit quantum state and a classical register $\reg{R^{\mathrm c}}$ that stores $|\reg{R^{\mathrm c}}|$ classical bits. A state $\rho$ stored across registers $\reg{S_1,\dots,S_k}$ is denoted by $\rho^{\reg{S_1,\dots,S_k}}$. We sometimes identify a quantum register with the Hilbert space of the quantum state it stores.
We will also abuse notation and write $(\rho, \sigma)$ to denote a \emph{possibly entangled state} on two subsystems. The maximally mixed state on a Hilbert space $\cH$ is the state $\frac{I_\cH}{\dim(\cH)}$, where $I_\cH$ is the identity operator on $\cH$.

\paragraph{Quantum channels.}
A quantum channel $\Phi$ is a completely-positive trace-preserving (CPTP) map. The quantum channel of applying an operator isometry $V$ refers to the map $\rho \mapsto V \rho V^\dag$. 
The partial trace $\tr_R$ over register $\reg{R}$ is the unique linear map from registers $\reg{R},\reg{S}$ to $\reg{S}$ such that $\tr_{R}(\rho^{\reg{R}} \ot \sigma^{\reg{S}}) = \tr(\rho) \sigma$ for every product state $\rho \ot \sigma$. 

It is well-known that (\eg see \cite{watrous2025understanding}) for any quantum channel $\Phi$ that maps $n$ qubits to $m$ qubits, there exists a unitary $U_\Phi$ on $(n+2m)$ qubits partitioned into registers $\reg{A},\reg{B}$ of $m, (n+m)$ qubits respectively, such that $\Phi$ can be written as 
\begin{align}
\label{eq:purification-of-channel}
    \Phi(\rho) = \tr_{B}\left[U_\Phi \left(\rho \otimes \ketbra{0}^{\ot 2m}\right) U_\Phi^\dag \right]
\end{align} 

\paragraph{Norm and Distance.} The trace norm  of a linear operator $A$ is defined as $\norm{A}_1 := \tr(\sqrt{A^\dagger A})$. The trace distance between two quantum states $\rho_0, \rho_1 \in \Den(\cH)$ is defined as $\frac12\norm{\rho_0-\rho_1}_1$. The diamond distance between two quantum maps $\cF,\cG: \Den(\cH) \to \Den(\cH')$ is defined as
$\| \cF- \cG \|_\diamond := \max_{\rho \in \Den(\cH^{\ot 2})} \| (\cF \otimes \identity_n) (\rho) - (\cG \otimes \identity_n) (\rho) \|_{1}$.
\begin{lemma}
\label{lemma:fidelity-trace-distance}
For states $\rho$ and $\phi$ where $\phi = \ketbra{\phi}$ is a pure state,
    \[\frac12 \norm{\rho - \phi}_1 \ge 1 - \bra{\phi}\rho\ket{\phi}\]
\end{lemma}

\paragraph{Haar measure.}
The Haar measure on the unit sphere of a Hilbert space $\cH$ is the unique probability measure that is invariant under action by every unitary map. The Haar measure on the group $U(\cH)$ of unitary operators over $\cH$ is the unique
probability measure $U(\cH)$ that is invariant under
left and right multiplication. A state/unitary sampled from the Haar measure is called a Haar random state/unitary.

\paragraph{Pauli and Clifford groups.}
The $n$-qubit Pauli group $\PauliGroup_n$ is the $n$-fold tensor of the group generated by  $X=\begin{pmatrix}0&1\\1&0\end{pmatrix}$ and $ Z=\begin{pmatrix}1&0\\0&-1\end{pmatrix}$. Its elements can be expressed as \[\pm i^{c} X^{a_1}Z^{b_1}\otimes\cdots\otimes X^{a_n}Z^{b_n} \]
for some $c, a_1,b_1,\dots,a_n,b_n \in \zo$. For a Hilbert space $\cH = (\bbC^2)^{\ot n}$, we write $\PauliGroup_{\cH} = \PauliGroup_n$.
The $n$-qubit Clifford group $\CliffordGroup_n$ is the normalizer of the Pauli group $\PauliGroup_n$.
$$
    \CliffordGroup_n
    :=
    \{n\text{-qubit unitary }E \text{ satisfying } E \PauliGroup_n E^\dagger = \PauliGroup_n\}.
$$

\begin{lemma}
[Pauli mixing]
\label{lemma:Pauli-mixing}
Let $\onreg{\rho}{A,B}$ be any state on quantum registers $\reg{A,B}$. Then we have
\[
\E_{P \gets \PauliGroup_{|B|}} \left[ \left(\onreg{I}{A} \ot \onreg{P}{B}\right) \onreg{\rho}{A,B} \left(\onreg{I}{A} \ot P^{\dag\;\reg{B}}\right) \right]
= \tr_{B}\left[\onreg{\rho}{A,B}\right] \ot \rho_{\mathsf{maxmixed}}^{\reg{B}},
\]
where $\rho_{\mathsf{maxmixed}}$ denotes the maximally mixed state on register $\reg{B}$.
\end{lemma}

\subsection{Teleportation}
We recall the ordinary quantum teleportation. Define EPR pair $\ket{\epr} := \frac{1}{\sqrt 2}\left(\left|00\right\rangle+\left|11\right\rangle\right)$.

\begin{lemma}[Quantum teleportation]
\label{lem:teleportation}
Let $\gray{M}$ be an $n$-qubit message register, $\gray{R}$ be an arbitrary auxiliary register, and registers $\gray{A},\gray{B}$ share the state
$\ketbra{\epr}^{\ot n}$. There is a measurement on $\gray{M,A}$ that, for every state
$\rho^{\gray{M,R}}$, produces strings $x,z\in\{0,1\}^n$ that are independent and
uniformly random. Moreover, conditioned on outcome $(x,z)$, the joint state on $\gray{B,R}$ is
\[
    \onreg{P_{x,z}}{B}\,
    \rho^{\gray{B,R}}\,
    \onreg{P_{x,z}^\dag}{B},
\]
where 
$P_{x,z} := X^{x_1}Z^{z_1} \otimes\cdots\otimes X^{x_n}
Z^{z_n} \in \PauliGroup_n$. Therefore, after applying the Pauli correction $P_{x,z}^\dag$ 
to register $\gray{B}$, the final joint state on $\gray{B,R}$ is exactly
$\rho^{\gray{B,R}}$.
\end{lemma}

\subsection{Port-based teleportation}
\label{sec:port-based-teleportation}

Port-based teleportation, introduced by Ishizaka and Hiroshima \cite{ishizaka2008asymptotic,ishizaka2009quantum}, is a variant of teleportation in which the receiver does not need to apply a Pauli correction to the output system.
In the beginning, the sender and receiver share many maximally entangled pairs, referred to as \emph{ports}. The sender performs one joint measurement on the input state and the ports held, obtains a classical label $i\in[N] := \{1,2,\dots,N\}$, and sends $i$ to the receiver. 
The receiver's $i$-th port would contain an approximation of the input state.

We use the following notation.  Let
$
    \reg{A_0 , A_i, B_i} \cong (\bbC^2)^{\ot m}
$
for every $i\in[N]$. Alice holds $\reg{A_0,A_1,\ldots,A_N}$, where
$\reg{A_0}$ is the input register, and Bob holds the ports $\reg{B_1,\ldots,B_N}$.  A
port-based teleportation measurement is a POVM
$
    \{M_i\}_{i=1}^N
$
on registers $\reg{A_0 A_1\cdots A_N}$.  If the measurement outcome is $i$, Bob keeps $\reg{B_i}$
and discards all other ports.

For an input state $\rho_{A_0R}$, possibly entangled with an arbitrary auxiliary
register $\reg{R}$, define the subnormalized state on the selected port $\reg{B_i}$ and
the auxiliary register $\reg{R}$ by
\begin{equation}
    \onreg{\tau^{(i)}}{B_i R}
    :=
    \tr_{{A_0 A_1\cdots A_N B_{\neq i}}}\!\left[
        \onreg{M_i}{A_0\cdots A_N }
        \left(
            \onreg{\rho}{A_0R}\otimes
            \bigotimes_{j=1}^N \onreg{\ketbra{\epr^{\ot m}}}{A_j B_j}
        \right)
    \right],
    \label{def:tau-i-for-PBT}
\end{equation}
where $\reg{B_{\neq i}}$ denotes all ports except $\reg{B_i}$. Thus
$\tr({\tau^{(i)}})$ is the probability of outcome $i$, and, conditioned on
outcome $i$, the selected output state is
$
    {\tau^{(i)}}/\tr(\tau^{(i)}).
$
We write $\widehat{\tau}^{(i)}$ for $\tau^{(i)}$ with $\reg{B_i}$ identified with
a common register $\reg{B_0}$, and we define the quantum channel $\mathcal{E}_{m,N}$ by \[\mathcal{E}_{m,N}(\rho) :=
    \sum_{i=1}^N \widehat{\tau}^{(i)}.\]

\begin{theorem}[Port-based teleportation {\cite{ishizaka2008asymptotic,ishizaka2009quantum}, Corollary II.2 of \cite{beigi2011simplified}}]
\label{thm:port-based-teleportation}
For every $m$ and every number of ports $N$, there is a POVM
$\{M_i\}_{i=1}^N$ such that the followings hold. 
\begin{itemize}
    \item Permutation-covariance: For every
    permutation $\pi\in S_N$, let $U_\pi$ be the unitary on
    $\reg{A_0A_1\cdots A_N}$ that fixes $\reg{A_0}$ and maps $\reg{A_j}$ to $\reg{A_{\pi(j)}}$, then we have
    \[
        U_\pi M_i U_\pi^\dagger = M_{\pi(i)}
        \qquad
        \text{for all } \pi\in S_N \text{ and } i\in[N].
    \]
    \item Correctness: The channel $\mathcal{E}_{m,N}$ induced by the POVM satisfies
    \[
        \norm{\mathcal{E}_{m,N}-\cI_{2^m}}_\diamond
        \le
        \frac{2^{2m+2}}{\sqrt N}.
    \]
\end{itemize}
\end{theorem}

\begin{lemma}[Extended correctness]
\label{lemma:extended-correctness-pbt}
Define the extended
channel $\widehat{\cE}_{m,N}$ by
\[
    \widehat{\cE}_{m,N}(\rho)
    :=
    \sum_{i=1}^N {\ketbra{i}}\otimes
    \widehat{\tau}^{(i)},
\]
where $\widehat{\tau}^{(i)}$ is defined according to \cref{def:tau-i-for-PBT} using the
POVM in \cref{thm:port-based-teleportation}.
Then for every auxiliary register
$\reg{R}$ and every input state $\onreg{\rho}{A_0R}$, we have 
\[\widehat{\cE}_{m,N}(\rho) = \frac{I_N}{N} \otimes \cE_{m,N}(\rho)\]
and therefore
\[
    \frac12\norm{
        \widehat{\cE}_{m,N}(\rho)
        -\frac{I_N}{N}\otimes{\rho}
    }_1
    \leq
    \frac{2^{2m+1}}{\sqrt N}.
\]
\end{lemma}

\begin{proof}
For a permutation $\pi\in S_N$, let $V_\pi$ permute Bob's ports by mapping
$\reg{B_j}$ to $\reg{B_{\pi(j)}}$. Since the state
\[\Psi := \onreg{\rho}{A_0 R}\otimes \bigotimes_{j=1}^N \onreg{\ketbra{\epr^{\ot m}}}{A_jB_j}\]
contains EPR states on registers $\reg{A_j B_j}$, it is invariant under the action of $\onreg{U_\pi}{A_0A_1\cdots A_N} \onreg{V_\pi}{B_1\cdots B_N}$. Then we have
\begin{align*}
    \widehat{\tau}^{\pi(i)}
    & =
    \tr_{{A_0 A_1\cdots A_N B_{\neq \pi(i)}}}\!\left[
        {M_{\pi(i)}}
        \Psi
    \right]\\
    & =
    \tr_{{A_0 A_1\cdots A_N B_{\neq \pi(i)}}}\!\left[
        \onreg{(U_\pi M_i U_\pi^\dagger)}{A_0\cdots A_N}
        \onreg{({U_\pi V_\pi}
        \Psi
        {V_\pi^\dagger U_\pi^\dagger})}{A_0\cdots A_N B_1\cdots B_N}
    \right]\\
    & =
    \tr_{{A_0 A_1\cdots A_N B_{\neq \pi(i)}}}\!\left[
        \onreg{M_i}{A_0\cdots A_N} \onreg{V_\pi}{B_1\cdots B_N}
        \Psi
        V_\pi^\dagger
    \right]\\
    & =
    \tr_{{A_0 A_1\cdots A_N B_{\neq i}}}\!\left[
        M_i
        \Psi
    \right] = \widehat{\tau}^{(i)}.
\end{align*}
Therefore, $\widehat{\tau}^{(i)}=\widehat{\tau}^{(j)}$ for all $i,j\in[N]$.
Recall that $\mathcal{E}_{m,N}(\rho) = \sum_i \widehat{\tau}^{(i)}$.
It follows that
\[
\begin{aligned}
    \widehat{\cE}_{m,N}(\rho)
    = 
    \sum_{i=1}^N \ketbra{i}\otimes\widehat{\tau}^{(i)}
    =
    I_N \otimes \frac{1}{N}\sum_j \widehat{\tau}^{(j)}
    =
    \frac{I_N}{N}\otimes
        \mathcal{E}_{m,N}(\rho).
\end{aligned}
\]
Combined with the correctness as in \cref{thm:port-based-teleportation}, we have
\[
\begin{aligned}
    \frac12\norm{
        \widehat{\cE}_{m,N}(\rho)
        -\frac{I_N}{N}\otimes\rho
    }_1
     =
    \frac12\norm{
        \frac{I_N}{N}\otimes
        \bigl(\mathcal{E}_{m,N}(\rho)-\rho\bigr)
    }_1
     \leq
    \frac12\norm{\mathcal{E}_{m,N}-\cI_{2^m}}_\diamond
    \leq
    \frac{2^{2m+1}}{\sqrt N}.
\end{aligned}
\]
\end{proof}

\subsection{Distributed inner product estimation}

\begin{dfn}[\cite{arunachalam2024distributed}]
In the decision DIPE$_{c,n}$ problem, Alice is given $\ket{\psi}^{\ot c}$ and Bob is given $\ket{\phi}^{\ot c}$. They are asked to decide which of the following two cases they are in.
\begin{itemize}
    \item YES: $\ket{\psi}=\ket{\phi}$ is the same $n$-qubit Haar random state
    \item NO: $\ket{\psi}, \ket{\phi}$ are independent $n$-qubit Haar random states.
\end{itemize}
A protocol is said to solve the decision DIPE$_{c,n}$ problem with correctness $\eta$ and soundness error $s$ if the protocol outputs $0$ in the YES case with probability $\ge \eta$ and outputs $0$ in the NO case with probability $\le s$. We simply say that a protocol solves the decision DIPE$_{c,n}$ problem if the correctness and soundness error has a constant gap $\eta - s = \Omega(1)$.
\end{dfn}

\begin{theorem}[\cite{arunachalam2024distributed}]
\label{thm:DIPE-lower-bound}
Any protocol that solves the decision DIPE$_{c,n}$ problem with $q$ qubits of quantum communication  and unlimited classical
communication requires $c = \Omega(\sqrt{2^{n-q}})$.
\end{theorem}

It is well known that, if the states $\ket{\psi},\ket{\phi}$ were held by the same party, one could run the swap test to decide the DIPE problem. The swap test channel is defined as \[\Phi_\SWAP(\psi,\phi) := \frac{1+\tr(\psi \phi)}{2} \ketbra{0} + \frac{1-\tr(\psi \phi)}{2} \ketbra{1}\]
$\Phi_\SWAP(\psi,\phi) = \ketbra{0}$ always holds in the YES case. $\E \left[\Phi_\SWAP(\psi,\phi)\right]$ is exponentially close to $\frac12 (\ketbra{0}+\ketbra{1})$ in the NO case, since independent Haar states satisfy $\E[ \tr(\psi \phi)] = 1/2^n$.

\section{Connecting quantum decomposable randomized encoding and quantum private simultaneous message}
\label{sec:connection}

We start by defining quantum decomposable randomized encoding, a cryptographic object that was first formalized by \cite{brakerski2022quantum}.
The QDRE definition in \cite{brakerski2022quantum} has a minor difference from  ours, where \cite{brakerski2022quantum} additionally introduced an offline encoding procedure that depends on the circuit $Q$ but not the input state. This was considered mainly for introducing the notion of universal QDRE, where the online encoding is made independent of $Q$. We focus only on QDRE for a public circuit $Q$, and can always include the offline encoding within one of the online encodings, with no impact on properties of QDRE. We work in the information-theoretic setting, where the running time of the algorithms and adversaries is unbounded.

\begin{dfn}
\label{dfn:QDRE}(QDRE)
    A quantum decomposable randomized encoding (QDRE) for a quantum circuit $Q$ consists of algorithms $(\Setup, \Enc,\Dec)$ with the syntax:
    \begin{itemize}
        \item $\Setup\rightarrow (g,\phi_1, \dots,\phi_k)$ outputs a classical string $g$ and a (possibly entangled) quantum state $(\phi_1, \dots,\phi_k)$ on setup registers $S_1,\dots,S_k$.
        \item $\Enc(\rho_i,i,g,\phi_i) \rightarrow \wt{\rho}_i = (g_i',\phi_i'), i\in [k]$ takes as input the $i$-th qubit $\rho_i$ of the quantum input, the index $i$, a string $g$, a quantum state $\phi_i$, and outputs an encoding $\wt{\rho}_i = (g_i',\phi_i')$ consisting of a classical part $g_i'$ and a quantum part $\phi_i'$.
        \item $\Dec(\wt{\rho_1},\dots,\wt{\rho_k})\rightarrow \rho_{\out}$ takes as input quantum encodings $\wt{\rho_1},\dots,\wt{\rho_k}$, and produces a quantum output $\rho_{\out}$.
    \end{itemize}
    It should satisfy the following properties:
    \begin{itemize}
        \item  $\epscor$-correctness-error: For every auxiliary register $R$ and every (possibly entangled) state $(\rho_1,\dots,\rho_k,\rho_\aux)$ with $\rho_\aux$ on $R$,
        \[\frac{1}{2}\norm{\E\left[\left(\Dec(\wt{\rho_1},\dots,\wt{\rho_k}), \rho_\aux\right)
          \,\middle|\,
      \begin{array}{l}
        (g, \phi_1, \dots, \phi_k) \gets \Setup \\
        \wt{\rho_i} \gets \Enc(\rho_i,i,g,\phi_i)
      \end{array}\right] - \left(Q(\rho_1, \dots, \rho_k), \rho_\aux\right)
    }_1 \le \epscor.
    \]
        \item  $\epspriv$-privacy-error: There exists a simulator algorithm $\Sim$ such that for every auxiliary register $R$ and every (possibly entangled) state $(\rho_1,\cdots,\rho_k,\rho_\aux)$ with $\rho_\aux$ on $R$,
$$ \frac12 \norm{ \E\left[(\wt{\rho}_1, \dots, \wt{\rho}_k, \rho_\aux)
              \,\middle|\,
          \begin{array}{l}
            (g, \phi_1, \dots, \phi_k) \gets \Setup \\
            \wt{\rho_i} \gets \Enc({\rho_i},i,g,\phi_i)
          \end{array}
        \right] - \left(\Sim\of{Q\of{\rho_1,\dots,\rho_k}}, \rho_\aux\right) }_1 \leq \epspriv. $$ 

    \end{itemize}
    We call $0$-correctness-error \emph{perfect correctness} and $0$-privacy-error \emph{perfect privacy}. A QDRE with perfect correctness and perfect privacy is called a perfect QDRE.
\end{dfn}

\begin{dfn}[Quantum encoding size]
    Given a QDRE $\boldsymbol{\mathsf{E}}$ where the encoding of the $i$-th input qubit is $\wt{\rho}_i=(g_i',\phi_i')$, we define its quantum encoding size as $ \mathsf{QES}(\boldsymbol{\mathsf{E}}) := \sum_{i\in[k]} |\phi'_i|$ and its setup-state size as $\mathsf{SS}(\boldsymbol{\mathsf{E}}) := \sum_{i\in[k]} |\phi_i|$.
\end{dfn}

\label{sec:comparison-by}

Next, we define a communication model, called quantum simultaneous message protocols, for computing a quantum channel $Q: \Den(\cH_1\ot\cdots\ot\cH_k) \to \Den(\cH_{\out})$ whose input Hilbert space is given as a tensor of $k$ parts explicitly.

\begin{dfn}[QSMP] \label{QSMP}
    Let $\cH_1,\dots,\cH_k,\cH_{\out}$ be Hilbert spaces. A $k$-party quantum simultaneous message protocol (QSMP), denoted by $\boldsymbol{\Pi}$, consists of local quantum channels $\Pi_i:\Den(\cH_i\ot S_i)\to \Den(M_i)$ with message registers $M_i$ for party $i\in[k],$ a pre-shared entangled state $\phi^\reg{S_1,\dots,S_k}\in \Den(S_1\ot\cdots\ot S_k)$ on pre-shared registers $S_i$ for party $i\in[k],$ and a quantum channel $Ref:\Den(M_1\ot\cdots\ot M_k)\to \Den(\cH_{\out})$ for the referee. In other words, $\boldsymbol{\Pi}=(\Pi_1,\dots,\Pi_k,\phi^\reg{S_1,\dots,S_k},Ref)$. On input state $\rho^\reg{\cH_1,\dots,\cH_k}\in \Den(\cH_1\ot\cdots\ot\cH_k),$ we denote the joint message state received by the referee as
    $$
    \View_{\boldsymbol{\Pi}}\of{\rho^\reg{\cH_1,\dots,\cH_k}}
    :=
    \of{\Pi_1\ot\cdots\ot\Pi_k}
    \of{
        \rho^\reg{\cH_1,\dots,\cH_k}\ot \phi^\reg{S_1,\dots,S_k}
    }.
    $$
    We say that $\boldsymbol{\Pi}$ computes a quantum channel $Q$ with $\epscor$-correctness-error if for every auxiliary register $R$ and every input state $\rho^\reg{\cH_1,\dots,\cH_k,R}\in \Den(\cH_1\ot\cdots\ot\cH_k \ot R)$,
        $$
        \frac12
        \norm{
            \left(Ref \circ \View_{\boldsymbol{\Pi}} \ot \cI_R\right)\of{\rho^\reg{\cH_1,\dots,\cH_k,R}}
            -
            (Q\ot \cI_R)\of{\rho^\reg{\cH_1,\dots,\cH_k, R}}
        }_1
        \leq
        \epscor.
        $$
    We say that it satisfies \emph{perfect correctness} if $\epscor = 0$.
\end{dfn}
\begin{dfn}[QPSM] \label{QPSM}
    A $(\epscor,\epspriv)$-quantum private simultaneous message protocol, or $(\epscor,\epspriv)$-QPSM protocol, for a quantum channel $Q$ is a QSMP $\boldsymbol{\Pi}$ that satisfies $\epscor$-correctness-error and $\epspriv$-privacy-error defined as follows.
    \begin{itemize}
        \item $\epscor$-correctness-error: For every auxiliary register $R$ and every input state $\rho^\reg{\cH_1,\dots,\cH_k,R}\in \Den(\cH_1\ot\cdots\ot\cH_k \ot R)$,
        $$
        \frac12
        \norm{
            \left(Ref \circ \View_{\boldsymbol{\Pi}} \ot \cI_R\right)\of{\rho^\reg{\cH_1,\dots,\cH_k,R}}
            -
            (Q\ot \cI_R)\of{\rho^\reg{\cH_1,\dots,\cH_k, R}}
        }_1
        \leq
        \epscor.
        $$
        \item $\epspriv$-privacy-error: 
        There exists a simulator $\Sim:\Den(\cH_{\out})\to \Den(M_1\ot\cdots\ot M_k) $ such that for every auxiliary register $R$ and every input state $\rho^\reg{\cH_1,\dots,\cH_k,R}\in \Den(\cH_1\ot\cdots\ot\cH_k \ot R)$,
$$ \frac12 \norm{ \left(\View_{\boldsymbol{\Pi}} \ot \cI_R\right)\of{\rho^\reg{\cH_1,\dots,\cH_k,R}} - \left(\Sim \ot \cI_R\right)\of{(Q \ot \cI_R)\of{\rho^\reg{\cH_1,\dots,\cH_k,R}}} }_1 \leq \epspriv. $$

    \end{itemize}
    We call $0$-correctness-error \emph{perfect correctness} and $0$-privacy-error \emph{perfect privacy}. A QPSM with perfect correctness and perfect privacy is called a perfect QPSM protocol. 
    When each of $\cH_1, \dots, \cH_k$ stores $n$ qubits, we call the protocol a QPSM$_{k,n}$ protocol.
\end{dfn}

\begin{dfn}[Communication cost]
    Given a QPSM protocol (or QSMP) $\boldsymbol{\Pi}$ with message registers $M_i = (M_i^{\mathrm q}, M_i^{\mathrm c})$ and pre-shared registers $S_i = (S_i^{\mathrm q}, S_i^{\mathrm c})$, under the notation in \cref{sec:prelim}, we define the following complexities.
    \begin{itemize}
        \item Quantum communication complexity
        $$
            \QCC(\boldsymbol{\Pi}) := \sum_{i\in[k]} |M_i^{\mathrm q}|.
        $$
        \item Entanglement complexity
        $$
            \mathsf{EC}(\boldsymbol{\Pi}) := \sum_{i\in[k]} |S_i^{\mathrm q}|.
        $$
        representing the number of pre-shared entangled qubits.
    \end{itemize}
\end{dfn}

\subsection{Connections}

\begin{theorem}[Bidirectional correspondence between QDRE and QPSM$_{k,1}$]
\label{thm:QDRE-QPSM-equivalence}
Let $Q$ be a quantum circuit with a $k$-qubit input, where each input register
consists of one qubit. Then $Q$ admits a QDRE scheme
$\boldsymbol{\mathsf{E}}=(\Setup,\Enc,\Dec)$ with $\epscor$-correctness-error and
$\epspriv$-privacy-error if and only if $Q$ admits a
$(\epscor,\epspriv)$-QPSM$_{k,1}$ protocol
$\boldsymbol{\Pi}=(\Pi_1,\ldots,\Pi_k,\phi^\reg{S_1,\ldots,S_k},Ref)$.
Moreover, the two transformations preserve the resource measure exactly:
$$
    \mathsf{QES}(\boldsymbol{\mathsf{E}})
    =
    \QCC(\boldsymbol{\Pi}) \text{ and } \mathsf{SS}(\boldsymbol{\mathsf{E}}) = \mathsf{EC}(\boldsymbol{\Pi}).
$$
\end{theorem}

We prove the theorem through the following \cref{lemma:QDRE_imply_k-QPSM,lemma:k-QPSM_imply_QDRE}.
The first lemma shows that a QDRE can be viewed as a $k$-party QPSM protocol in which each party holds one qubit. 
A QDRE can also be viewed as a two-party QPSM protocol, as given below in \cref{cor:QDRE_imply_2-QPSM}.

\begin{lemma}[QDRE implies QPSM$_{k,1}$] \label{lemma:QDRE_imply_k-QPSM}
Let $Q$
be a quantum circuit with a $k$-qubit input. Assume $Q$ admits a QDRE scheme $\boldsymbol{\mathsf{E}}=(\Setup,\Enc,\Dec)$ with $\epscor$-correctness-error and $\epspriv$-privacy-error.
Then there is a $(\epscor,\epspriv)$-QPSM$_{k,1}$ protocol $\boldsymbol{\Pi}=(\Pi_1,\ldots,\Pi_k,\phi^\reg{S_1,\ldots,S_k},Ref)$ for $Q$ such that
$$
    \mathsf{QES}(\boldsymbol{\mathsf{E}})
    =
    \QCC(\boldsymbol{\Pi}) \text{ and } \mathsf{SS}(\boldsymbol{\mathsf{E}}) = \mathsf{EC}(\boldsymbol{\Pi}).
$$
\end{lemma}

\begin{proof}
We construct a QPSM protocol as follows. The pre-shared state is generated by $ \Setup\to (g,\varphi_1,\ldots,\varphi_k),$ and party $i$ receives $\phi^\reg{S_i} := (g,\varphi_i)$.
For each party $i\in[k]$ who receives a $1$-qubit input $\rho_i$, define its QPSM message as
$$ \Pi_i(\rho_i,(g,\varphi_i)) := \Enc(\rho_i,i,g,\varphi_i),$$
which is precisely the encoding $\widetilde\rho_i$. Then, we define the referee as
$$ Ref(M_1,\dots,M_k) := \Dec(M_1,\dots,M_k). $$
This specifies the QPSM protocol $\boldsymbol{\Pi}=(\Pi_1,\ldots,\Pi_k,\phi^\reg{S_1,\ldots,S_k},Ref)$.

To prove $\epscor$-correctness-error, fix an arbitrary auxiliary register $R$ and an arbitrary joint input state
$\rho=(\rho_1,\ldots,\rho_k,\rho_\aux)$. The joint message and auxiliary state is
$$\left(\View_{\boldsymbol{\Pi}}\ot\cI_R\right)\of{\rho}
    =
    (\Pi_1\otimes\cdots\otimes\Pi_k\ot\cI_R)
    \of{\rho\otimes\phi^\reg{S_1,\ldots,S_k}}
    =
    (\widetilde\rho_1,\ldots,\widetilde\rho_k,\rho_\aux),
$$
where the last expression denotes the joint encoded state produced by the QDRE encoding algorithms together with the auxiliary register.
Therefore,
\begin{align*}
    &\frac12
    \norm{
        \left(Ref\circ\View_{\boldsymbol{\Pi}}\ot\cI_R\right)
        \of{\rho}
        -
        \left(Q\ot\cI_R\right)\of{\rho}
    }_1\\
    =&
    \frac12
    \norm{
        \left(\Dec\of{\widetilde\rho_1,\ldots,\widetilde\rho_k},\rho_\aux\right)
        -
        \left(Q\ot\cI_R\right)\of{\rho}
    }_1
    \leq \epscor,
\end{align*}
where the inequality is due to the $\epscor$-correctness-error of $\boldsymbol{\mathsf{E}}$. This proves correctness.

For $\epspriv$-privacy-error, by the same reasoning as above,
$\left(\View_{\boldsymbol{\Pi}}\ot\cI_R\right)
\of{\rho}$ is the joint QDRE encoding
together with its possibly entangled auxiliary register.
By $\epspriv$-privacy-error of $\boldsymbol{\mathsf{E}}$, there is a simulator
$\Sim$ satisfying
\begin{align*}
    &\frac12 \norm{
        \left(\View_{\boldsymbol{\Pi}}\ot\cI_R\right)
        \of{\rho}
        -
        \left(\Sim\ot\cI_R\right)
        \of{\left(Q\ot\cI_R\right)\of{\rho}}
    }_1 \\
    =& \frac12 \norm{
        \left(\widetilde\rho_1,\ldots,\widetilde\rho_k,\rho_\aux\right)
        -
        \left(\Sim(Q\of{\rho_1,\ldots,\rho_k}),\rho_\aux\right)
    }_1
    \leq \epspriv,
\end{align*}
which is exactly the $\epspriv$-privacy-error condition for QPSM.
Finally, since the quantum message sent by party $i$ is the state $\phi_i'$, we have
$$
    \QCC(\boldsymbol{\Pi})
    =
    \sum_{i\in[k]} |M_i^{\mathrm q}|
    =
    \sum_{i\in[k]} |\phi_i'|
    =
    \mathsf{QES}(\boldsymbol{\mathsf{E}}).
$$

$$
    \mathsf{EC}(\boldsymbol{\Pi})
    =
    \sum_{i\in[k]} |S_i^{\mathrm q}|
    =
    \sum_{i\in[k]} |\varphi_i|
    =
    \mathsf{SS}(\boldsymbol{\mathsf{E}}).
$$
\end{proof}

\begin{cor}[QDRE implies QPSM$_{2,n}$] \label{cor:QDRE_imply_2-QPSM}
Let $\cH_A,\cH_B \cong (\bbC^2)^{\otimes n}$ and consider a quantum circuit $Q:\Den(\cH_A\otimes\cH_B)\to\Den(\cH_{\out})$.
Assume $Q$ admits a QDRE scheme $\boldsymbol{\mathsf{E}}=(\Setup,\Enc,\Dec)$ with $\epscor$-correctness-error and $\epspriv$-privacy-error.
Then there is a $(\epscor,\epspriv)$-QPSM$_{2,n}$ protocol $\boldsymbol{\Pi}=(\Pi_A,\Pi_B,\phi^\reg{S_A,S_B},Ref)$ for $Q$ such that
$$
    \mathsf{QES}(\boldsymbol{\mathsf{E}})
    =
    \QCC(\boldsymbol{\Pi}) \text{ and } \mathsf{SS}(\boldsymbol{\mathsf{E}}) = \mathsf{EC}(\boldsymbol{\Pi}).
$$
\end{cor}

\begin{proof}
View $Q$ as a quantum circuit that takes a $2n$-qubit input. Applying \cref{lemma:QDRE_imply_k-QPSM} to $Q$ with $k=2n$ gives a $(\epscor,\epspriv)$-QPSM$_{2n,1}$ protocol
$$
    \boldsymbol{\Gamma}
    =
    (\Gamma_1,\ldots,\Gamma_{2n},\phi_{S_1,\ldots,S_{2n}},Ref_{\Gamma})
$$
for $Q$ such that
$$
    \mathsf{QES}(\boldsymbol{\mathsf{E}}) = \QCC(\boldsymbol{\Gamma}) \text{ and } 
    \mathsf{SS}(\boldsymbol{\mathsf{E}}) = \mathsf{EC}(\boldsymbol{\Gamma}).
$$

We group the first $n$ parties of $\boldsymbol{\Gamma}$ into Alice and the last $n$ parties into Bob.
Set $S_A=(S_1,\ldots,S_n)$, $S_B=(S_{n+1},\ldots,S_{2n})$, and let $\phi^\reg{S_A,S_B}:=\phi_{S_1,\ldots,S_{2n}}$.
On input $\rho_A\in\Den(\cH_A)$, Alice applies
$
    \Gamma_1\otimes\cdots\otimes\Gamma_n
$
to her $n$ input qubits and her registers $S_A$, and sends the resulting tuple of messages
$
    M_A=(M_1,\ldots,M_n).
$
Similarly, on input $\rho_B\in\Den(\cH_B)$, Bob applies
$
    \Gamma_{n+1}\otimes\cdots\otimes\Gamma_{2n}
$
to his $n$ input qubits and his registers $S_B$, and sends
$
    M_B=(M_{n+1},\ldots,M_{2n}).
$
Finally, define the referee by
$
    Ref(M_A,M_B):=Ref_{\Gamma}(M_1,\ldots,M_{2n}).
$
This gives a two-party protocol
$
    \boldsymbol{\Pi}=(\Pi_A,\Pi_B,\phi^\reg{S_A,S_B},Ref).
$

Fix an arbitrary auxiliary register $R$ and an arbitrary joint
input state $\rho=(\rho_A,\rho_B,\rho_\aux)\in
\Den(\cH_A\ot\cH_B\ot R)$. The referee's view in
$\boldsymbol{\Pi}$, including the untouched auxiliary register, is exactly the
corresponding view in $\boldsymbol{\Gamma}$, \ie
\begin{align*}
    \left(\View_{\boldsymbol{\Pi}}\ot\cI_R\right)
    \of{\rho}
    &=
    \left(\View_{\boldsymbol{\Gamma}}\ot\cI_R\right)
    \of{\rho}.
\end{align*}
Therefore, the $\epscor$-correctness-error and $\epspriv$-privacy-error of $\boldsymbol{\Pi}$ follow from those of $\boldsymbol{\Gamma}$.
Thus $\boldsymbol{\Pi}$ is a $(\epscor,\epspriv)$-QPSM$_{2,n}$ protocol for $Q$.
Finally, grouping messages does not change the total number of qubits, so
$
    \QCC(\boldsymbol{\Pi})
    =
    \QCC(\boldsymbol{\Gamma})
    =
    \mathsf{QES}(\boldsymbol{\mathsf{E}})
$ and $\mathsf{EC}(\boldsymbol{\Pi})
    = \mathsf{EC}(\boldsymbol{\Gamma})
    =
    \mathsf{SS}(\boldsymbol{\mathsf{E}})$.
\end{proof}

By \cref{cor:QDRE_imply_2-QPSM}, any two-party QPSM lower bound on QCC implies a QDRE lower bound on QES, allowing us to study QDRE lower bounds through the setting of two-party QPSM lower bounds. The converse of the corollary need not be true, but we will show in the following that the converse of \cref{lemma:QDRE_imply_k-QPSM} does hold. Hence, one can obtain QDRE upper bounds through studying QPSM$_{k,1}$ upper bounds.

\begin{lemma}[QPSM$_{k,1}$ implies QDRE]
\label{lemma:k-QPSM_imply_QDRE}
Let $Q$ be a quantum circuit with a $k$-qubit input.
Suppose there exists a $(\epscor,\epspriv)$-QPSM$_{k,1}$ protocol $\boldsymbol{\Pi}=(\Pi_1,\ldots,\Pi_k,\phi^\reg{S_1,\ldots,S_k},Ref)$ for $Q$.
Then $Q$ admits a QDRE scheme $\boldsymbol{\mathsf{E}}=(\Setup,\Enc,\Dec)$ with $\epscor$-correctness-error and $\epspriv$-privacy-error such that
$$
    \mathsf{QES}(\boldsymbol{\mathsf{E}})=\QCC(\boldsymbol{\Pi}) \text{ and } \mathsf{SS}(\boldsymbol{\mathsf{E}}) = \mathsf{EC}(\boldsymbol{\Pi}).
$$
\end{lemma}

\begin{proof}
We construct a QDRE scheme as follows.

\begin{itemize}
    \item $\Setup \to (\emptyset, \phi_1,\dots,\phi_k)$ where $(\phi_1,\dots,\phi_k) = \phi^\reg{S_1,\dots,S_k}$.
    \item $\Enc(\rho_i,i,g,\phi_i) \to \wt{\rho}_i = \Pi_i(\rho_i, \phi_i)$.
    \item $\Dec(\wt{\rho_1},\dots,\wt{\rho}_k) \to Ref(\wt{\rho_1},\dots,\wt{\rho}_k)$.
\end{itemize}

We first prove $\epscor$-correctness-error. Fix an arbitrary auxiliary register $R$ and an arbitrary joint state
$\rho=(\rho_1,\ldots,\rho_k,\rho_\aux)$. The joint encoding generated by the
QDRE, together with the untouched auxiliary register, is
\begin{align*}
    (\wt{\rho}_1,\ldots,\wt{\rho}_k,\rho_\aux)
    =& (\Pi_1(\rho_1, \phi_1),\ldots,\Pi_k(\rho_k, \phi_k),\rho_\aux)\\
    =&\of{\Pi_1\ot\cdots\ot\Pi_k\ot\cI_R}
    \of{
        \rho \ot \phi^\reg{S_1,\dots,S_k}
    }
    = \left(\View_{\boldsymbol{\Pi}}\ot\cI_R\right)\of{\rho}.
\end{align*}
Thus the QDRE output satisfies
\begin{align*}
    &\frac12
    \norm{
        \left(\Dec(\wt{\rho_1},\ldots,\wt{\rho_k}),\rho_\aux\right)
        -
        \left(Q\ot\cI_R\right)\of{\rho}
    }_1\\
    =&
    \frac12
    \norm{
        \left(Ref\circ\View_{\boldsymbol{\Pi}}\ot\cI_R\right)
        \of{\rho}
        -
        \left(Q\ot\cI_R\right)\of{\rho}
    }_1
    \leq \epscor,
\end{align*}
where the inequality follows from the $\epscor$-correctness-error of $\boldsymbol{\Pi}$.

We next prove $\epspriv$-privacy-error. 
By the $\epspriv$-privacy-error of $\boldsymbol{\Pi}$, there is a simulator
$\Sim$ such that for every auxiliary register $R$ and every joint state $\rho=(\rho_1,\ldots,\rho_k,\rho_\aux)$,
\begin{align*}
    &\frac12
    \norm{
        \left(\View_{\boldsymbol{\Pi}}\ot\cI_R\right)
        \of{\rho}
        -
        \left(\Sim\ot\cI_R\right)
        \of{\left(Q\ot\cI_R\right)\of{\rho}}
    }_1\\
    =&
    \frac12
    \norm{
        (\wt{\rho}_1,\ldots,\wt{\rho}_k,\rho_\aux)
        -
        \left(\Sim\of{Q\of{\rho_1,\ldots,\rho_k}},\rho_\aux\right)
    }_1
    \leq \epspriv,
\end{align*}
Finally, by construction, the quantum part of the QDRE encoding of the $i$-th qubit is exactly the quantum part of $\Pi_i(\rho_i,\phi_i)$. Therefore, we have $|\phi_i'|=|M_i^{\mathrm q}|$ for every $i\in[k]$, and hence
$$
    \mathsf{QES}(\boldsymbol{\mathsf{E}})
    =
    \sum_{i\in[k]}|\phi_i'|
    =
    \sum_{i\in[k]}|M_i^{\mathrm q}|
    =
    \QCC(\boldsymbol{\Pi}).
$$
$$
    \mathsf{SS}(\boldsymbol{\mathsf{E}})
    =
    \sum_{i\in[k]}|\phi_i|
    =
    \sum_{i\in[k]}|S_i^{\mathrm q}|
    =
    \mathsf{EC}(\boldsymbol{\Pi}).
$$

\end{proof}

\section{Quantum communication complexity}

We study the quantum communication complexity $\QCC$ of quantum private simultaneous message (QPSM) protocols with pre-shared entanglement, where the parties wish to minimize the amount of quantum information sent to the referee while they can transmit classical information to the referee for free. 

\subsection{An output-size upper bound for $k$-party QPSM}

\label{subsec:multiparty-1QCC}

Suppose that the parties wish to compute a quantum channel $\Phi$ that takes an input state of length $n = \sum_i n_i$, with party $i$ providing $n_i$ qubits, and produces an output of length $m$ consisting of $m_c$ classical bits and $m_q$ qubits. We give a QPSM protocol with quantum communication complexity $m_q$, which is $1$ for $1$-qubit outputs and $0$ for classical outputs. 

\begin{theorem}
\label{thm:upperbound}
    For every quantum channel $\Phi: \Den(\cH_1\ot \cdots \ot \cH_k) \to \Den((\bbC^2)^{\otimes m_q}; \{0,1\}^{m_c}) $ that outputs $m_c$ classical bits and $m_q$ qubits, and every constants $\epscor,\epspriv > 0$, there exists a $k$-party $(\epscor,\epspriv)$-QPSM protocol with quantum communication complexity $\QCC = m_q$.
\end{theorem}

The $k$-party QPSM protocol is described as follows, where party $i$ receives input $\rho_i$ and pre-shares the following entangled states with $t \in [N]^0 \cup \cdots \cup [N]^{k-1}$.

\begin{itemize}
    \item Let $(e^{i \to 1}_i, e^{i \to 1}_1)$ be $n_i$ EPR pairs pre-shared between party $i$ and party $1$.
    \item Let $(\tau^{i \to (i+1)}_{i,t}, \tau^{i \to (i+1)}_{i+1,t})$ be $n$ EPR pairs pre-shared between party $i$ and party $i+1$.
    \item Let $(g^{(i+1) \to i}_{i+1,t}, g^{(i+1) \to i}_{i,t})$ be $m$ EPR pairs pre-shared between party $i+1$ and party $i$.
\end{itemize}

\begin{protocol}{}
\label{protocol:QPSM:1QCC}
Set $\epsport
    :=
    \min\left\{
        \epscor,
        \epspriv
    \right\}/ k$ and $N = \Theta\left(2^{4n}/{\epsport^2}\right)$.
\begin{itemize}
    \item Every party $i \neq 1$ applies ordinary teleportation to their input $\rho_i$ using half EPR $e^{i\to 1}_i$, producing Pauli correction $P_i$.
    \item Party $1$ applies port-based teleportation to the single joint state $(\rho_1, e^{2\to 1}_1, \dots, e^{k\to 1}_1)$ using the $N$ ports $\tau^{1 \to 2}_{1,1},\dots,\tau^{1 \to 2}_{1,N}$, producing a port index $x_{1} \in [N]$.
    \item Every party $i \neq 1, k$ enumerates all tuples $t \in [N]^{i-1}$, applies the Pauli correction $(I \cdots \otimes P_i^\dag \otimes \cdots I)$ to the state $\tau^{(i-1) \to i}_{i,t}$, and applies port-based teleportation to the resulting state using the ports $\tau^{i \to (i+1)}_{i,(t,1)},\dots,\tau^{i \to (i+1)}_{i,(t,N)}$, producing a port index $x_{i,t} \in [N]$.
    \item Party $k$ enumerates all tuples $t \in [N]^{k-1}$, applies Pauli correction $(I \cdots \otimes P_k^\dag)$ to the state $\tau^{(k-1) \to k}_{k,t}$ followed by applying the computation $\Phi$, and applies ordinary teleportation to the resulting state using half EPR $g^{k \to (k-1)}_{k,t}$, producing Pauli correction $P'_{k,t}$.
    \item Every party $i \neq 1, k$ enumerates all tuples $t \in [N]^{i-1}$ and applies ordinary teleportation to the state $g^{(i+1) \to i}_{i,(t,x_{i,t})}$ using half EPR $g^{i \to (i-1)}_{i, t}$, producing Pauli correction $P'_{i,t}$.
    \item Party $1$ sends $(x_1, g^{2 \to 1}_{1,x_1})$, party $i \neq 1,k$ sends $(x_{i,t}, (P'_{i,t})_{t \in [N]^{i-1}})$, and party $k$ sends $(P'_{k,t})_{t \in [N]^{k-1}}$ to the referee.
    \item The referee recursively computes $y_1 \gets x_1$ and $y_{i} \gets (y_{i-1}, x_{i,y_{i-1}})$ for $1<i < k$, sets $P^\star:=P'_{2,y_1}\cdots P'_{k,y_{k-1}}$, applies $(P^\star)^\dag$ to $g^{2 \to 1}_{1,x_1}$, and outputs the resulting state.
\end{itemize}
\end{protocol}

Note that party $1$ should treat $g^{2\to 1}_{1,x_1}$ as the output of the computation $\Phi$, in the sense that the indices that are expected to be classical bits according to the description of $\Phi$ should be measured and sent as classical bits. This does not affect the final outcome since measurement in the standard basis commutes with Pauli corrections up to classical postprocessing. The overall quantum communication complexity of the protocol is $m_q$ qubits.

\begin{lem}\label{lemma:correctness-and-privacy:QPSM:1QCC}
\cref{protocol:QPSM:1QCC} satisfies $\epscor$-correctness-error and $\epspriv$-privacy error.
\end{lem}

\begin{proof}
First, we give an argument without an auxiliary register. Let
$\rho=(\rho_1,\ldots,\rho_k)$ be an arbitrary (possibly entangled) joint input state.
Let $\View$ denote the induced view channel. Thus
$\View(\rho_1,\ldots,\rho_k)$ consists of the
quantum state $g^{2\to1}_{1,x_1}$ and classical strings
\[
    s = \left(\{x_{i,t}\}_{i<k,t\in[N]^{i-1}},
    \{P'_{i,t}\}_{i>1,t\in[N]^{i-1}}\right).
\]
We define a simulator $\Sim$ that takes
$\sigma:=\Phi(\rho_1,\ldots,\rho_k)$ as input. It samples all port indices
$x_{i,t}$ independently and uniformly from $[N]$ and all Pauli corrections
$P'_{i,t}$ independently and uniformly. It sets $y_1=x_1$ and
$y_i=(y_{i-1},x_{i,y_{i-1}})$ recursively for $i>1$, and sets
$P^\star := P'_{2,y_1}P'_{3,y_2}\cdots P'_{k,y_{k-1}}$. The simulator outputs the quantum state $P^\star \sigma P^{\star \dag}$ along with classical strings $\left(\{x_{i,t}\}_{i<k,t\in[N]^{i-1}}, \{P'_{i,t}\}_{i>1,t\in[N]^{i-1}}\right)$.
Below, we show that
\[
    \frac12 \norm{
        \View(\rho_1,\ldots,\rho_k)-\mathsf{Sim}(\Phi(\rho_1,\ldots,\rho_k))
    }_1
    \le \min\{\epspriv, \epscor\}.
\]
Let $P_i$ be the Pauli correction that party $i$ learns as described in the
protocol. Define
\[
    \rho^{(r)}
    :=
    \left(
        I_1\otimes\cdots\otimes I_{r}
        \otimes
        P_{r+1}\otimes\cdots\otimes P_k
    \right)
    \rho
    \left(
        I_1\otimes\cdots\otimes I_{r}
        \otimes
        P_{r+1}\otimes\cdots\otimes P_k
    \right)^\dag.
\]
By correctness of ordinary teleportation, party $1$ holds the following state after the initial ordinary teleportations from parties $i\neq 1$:
\[
    (\rho_1, e^{2\to1}_1, \dots, e^{k\to1}_1) = \left(
        I_1\otimes P_{2}\otimes\cdots\otimes P_k
    \right) \rho \left(
        I_1\otimes P_{2}\otimes\cdots\otimes P_k
    \right)^\dag
    = \rho^{(1)}.
\]
For $j\geq 1$, write
$\mu_j:=(I_N/N)^{\otimes N^{j-1}}$.
We prove by induction on $r=1,\ldots,k-1$ that
\[
    \frac12
    \norm{
        \E\left(x_1, \dots, x_r,
        \tau^{r\to (r+1)}_{r+1,y_r}\right)
        -
        \left(\mu_1, \ldots, \mu_r, \rho^{(r)}\right)
    }_1
    \leq
    r\epsport,
\] 
where we denote the vector $x_j = (x_{j,t})_{t \in [N]^{j-1}}$, the sequence $y_1 = x_1$ and $y_{i} = (y_{i-1}, x_{i,y_{i-1}})$, and the expectation is taken over the measurements of quantum states. 
For $r=1$, party $1$ applies port-based teleportation to $\rho^{(1)}$ and obtains measurement result $y_1=x_1$. Notice that $\E(x_1,\tau^{1\to 2}_{2,y_1}) = \widehat{\cE}_{n,N}\left(\rho^{(1)}\right)$, where $\widehat{\cE}_{n,N}$ was introduced in \cref{sec:port-based-teleportation}.
By the extended correctness of port-based teleportation (\cref{lemma:extended-correctness-pbt}) and the choice of $N$,
\[
    \frac12
    \norm{
        \E\left(x_1,\tau^{1\to 2}_{2,y_1}\right)
        -
        \left(\mu_1, \rho^{(1)}\right)
    }_1
    \le \frac{2^{2n+2}}{\sqrt{N}} \leq
    \epsport .
\]
Now assume the claim holds for some $r = i-1 \ge 1$.
Party $i$ applies the Pauli correction
\[
    \bar{P}_i^\dag := I_1\otimes\cdots\otimes I_{i-1}
    \otimes P_{i}^\dag
    \otimes I_{i+1}\otimes\cdots\otimes I_k
\]
to all received ports 
$\tau^{(i-1)\to i}_{i,t}$, and then applies port-based teleportation with
measurement result $x_{i,t}$. For $t=y_{i-1}$, extended correctness applies
to the selected port and gives
\[
    \E\left(x_{i,y_{i-1}},
        \tau^{i\to(i+1)}_{i+1,(y_{i-1},x_{i,y_{i-1}})}\right)
    =
    \widehat{\cE}_{n,N}
    \left(
        \bar P_i^\dag
        \tau^{(i-1)\to i}_{i,y_{i-1}}
        \bar P_i
    \right).
\]
By the extended correctness of port-based teleportation (\cref{lemma:extended-correctness-pbt}),
\[
    \E\left(x_i,\;
        \tau^{i\to(i+1)}_{i+1,(y_{i-1},x_{i,y_{i-1}})}\right)
    =
    \left(
        \mu_i,
        \cE_{n,N}
        \left(
            \bar P_i^\dag
            \tau^{(i-1)\to i}_{i,y_{i-1}}
            \bar P_i
        \right)
    \right)
\]
and therefore
\[\frac12
    \norm{
        \E\left(x_{i},\;
        \tau^{i\to(i+1)}_{i+1,(y_{i-1},x_{i,y_{i-1}})}\right)
        -
        \left(\mu_i, \bar{P}_i^\dag\;
        \tau^{(i-1)\to i}_{i,y_{i-1}}\;
        \bar{P}_i\right)
    }_1 \le \frac{2^{2n+2}}{\sqrt{N}} \le \epsport.\]
Combined with the triangle inequality, the induction hypothesis, and the relation $\rho^{(i)} = \bar{P}_i^\dag\;
        \rho^{(i-1)}\;
        \bar{P}_i$, we complete the induction:
\begin{align*}
    & \frac12
    \norm{
        \E\left(x_1, \dots, x_{i},
        \tau^{i\to (i+1)}_{i+1,y_i}\right)
        -
        \left(\mu_1, \dots, \mu_i, \rho^{(i)}\right)
    }_1 \\
    \le & 
    \frac12
    \norm{
        \E\left(x_1, \dots, x_{i},
        \tau^{i\to (i+1)}_{i+1,y_i}\right)
        -
        \E\left(x_1, \dots, x_{i-1},\mu_i,\bar{P}_i^\dag\;
        \tau^{(i-1)\to i}_{i,y_{i-1}}\;
        \bar{P}_i\right)
    }_1\\
    &\;+ \frac12 \norm{
        \E\left(x_1, \dots, x_{i-1},\mu_i,\bar{P}_i^\dag\;
        \tau^{(i-1)\to i}_{i,y_{i-1}}\;
        \bar{P}_i\right)
        -
        \left(\mu_1, \dots, \mu_i, \rho^{(i)}\right)
    }_1 \\
    \le & \epsport + \frac12 \norm{
        \E\left(x_1, \dots, x_{i-1},
        \mu_i,
        \tau^{(i-1)\to i}_{i,y_{i-1}}\right)
        -
        \left(\mu_1, \dots, \mu_i, \rho^{(i-1)}\right)
    }_1 \\
    \le & \epsport + (i-1) \epsport
    \le i \epsport.
\end{align*}
Then, party $k$ applies $\bar{P}_k^\dag$ and $\Phi$ to port
$\tau^{(k-1)\to k}_{k,t}$ for every $t$. By the contractivity of trace
distance and the relation
$\Phi(\rho) = \Phi(\bar{P}_k^\dag \rho^{(k-1)} \bar{P}_k)$, we have
\begin{equation}
\begin{aligned}[b]
    & \frac12
    \norm{
        \E\left(x_1, x_{2},\dots, x_{k-1},\Phi\left(
            \bar{P}_k^\dag
            \tau^{(k-1)\to k}_{k,y_{k-1}}
            \bar{P}_k
        \right)\right)
        -
        \left(\mu_1, \dots, \mu_{k-1}, \Phi(\rho)\right)
    }_1\\
    \leq &
    \frac12
    \norm{
        \E\left(x_1, x_{2},\dots, x_{k-1},
            \tau^{(k-1)\to k}_{k,y_{k-1}}\right)
        -
        \left(\mu_1, \dots, \mu_{k-1}, \rho^{(k-1)}\right)
    }_1 
    \leq
    (k-1)\epsport.
    \label{eq:a-bound-for-1QCC}
\end{aligned}
\end{equation}
The state $\Phi\left(
            \bar{P}_k^\dag
            \tau^{(k-1)\to k}_{k,y_{k-1}}
            \bar{P}_k
        \right)$ is teleported from party $k$ using $g^{k \to (k-1)}_{k,y_{k-1}}$ to party $k-1$. By the correctness of teleportation (\cref{lem:teleportation}), party $k-1$ receives the state on $g^{k \to (k-1)}_{k-1,y_{k-1}}$ with Pauli error $P'_{k,y_{k-1}}$ that party $k$ obtains from the teleportation measurement, which is a uniformly random Pauli. Iteratively, each party $i$ receives the state (up to Pauli errors) on $g^{(i+1) \to i}_{i,y_{i}}$ before teleporting the state further on. By induction, at the end of this process, party $1$ receives the state on $g^{2\to 1}_{1,x_1}$ with accumulated Pauli error $P'_{2,y_1}\cdots P'_{k,y_{k-1}}$, where all individual Pauli terms are independent and uniformly random.  

Define a channel $\mathcal{T}$ that, given the registers
$x_1,\ldots,x_{k-1}$ and a state $\sigma$, samples all
$P'_{i,t}$ independently and uniformly, records them in classical registers,
and replaces $\sigma$ by
\[
    P^\star\sigma P^{\star\dag},
    \qquad
    P^\star=P'_{2,y_1}P'_{3,y_2}\cdots P'_{k,y_{k-1}},
\]
where the indices $y_i$ are computed from the $x_i$ as above. By the above analysis, we have
\[\View(\rho) = \E\left[ \cT\left(x_1, x_{2},\dots, x_{k-1},\Phi\left(
            \bar{P}_k^\dag
            \tau^{(k-1)\to k}_{k,y_{k-1}}
            \bar{P}_k
        \right)\right) \right]\]
\[\Sim(\Phi(\rho))
= \cT\left(\mu_1, \dots, \mu_{k-1}, \Phi(\rho)\right)  \]
Contractivity of trace distance and inequality (\ref{eq:a-bound-for-1QCC}) then gives the desired inequality
\begin{align*}
    \frac12\norm{
        \View(\rho)-\Sim(\Phi(\rho))
    }_1
    &\le (k-1)\epsport
     \le \min\{\epspriv,\epscor\}.
\end{align*}

For the general case with an arbitrary auxiliary register $R$, the same calculation applies.
Indeed, the bounds for both ordinary teleportation and port-based teleportation retain
even when the input is entangled with an auxiliary register. In other words, for every joint state $\rho=(\rho_1,\dots,\rho_k,\rho_\aux)$, the above analysis gives
\[
    \frac12\norm{
        \left(\View\otimes\cI_R\right)\of{\rho} - \left(\Sim\otimes\cI_R\right)\of{(\Phi\otimes\cI_R)(\rho)}
    }_1
    \leq \min\{\epspriv,\epscor\}.
\]

The $\epspriv$-privacy-error follows immediately from the above inequality. The $\epscor$-correctness-error follows by the fact that the referee algorithm outputs $P^{\star^\dag} g^{2\to 1}_{1,x_1} P^\star = \Phi(\rho_1,\dots, \rho_k)$ on any simulated view sampled from $\mathsf{Sim}(\Phi(\rho_1,\dots, \rho_k))$, and therefore,
\begin{align*}
    & \frac12 \norm{
        ((Ref\circ\mathsf{View})\otimes\cI_R)\of{\rho} - \left(\Phi\otimes\cI_R\right)\of{\rho}
    }_1 \\
    = & \frac12 \norm{
        ((Ref\circ\View)\otimes\cI_R)\of{\rho} - \left((Ref\circ \Sim \circ \Phi)\otimes\cI_R\right)\of{\rho}
    }_1 \\
    \leq & \frac12 \norm{
        (\View\otimes\cI_R)\of{\rho} - ((\Sim \circ \Phi)\otimes\cI_R)\of{\rho}
    }_1 \leq \epscor.
\end{align*}

\end{proof}

\subsection{A separation between QPSM and QSMP}
We show a separation on the quantum communication complexity between QPSM and QSMP, thereby demonstrating how the privacy requirement affects the quantum communication complexity. Our separation considers the following computation of distributed generation of BB84 state: Alice receives $x\in\{0,1\}$, Bob receives
$\theta\in\{0,1\}$, and the referee Charlie outputs a single qubit $\ket{\psi_{x,\theta}}:= H^\theta \ket{x}$.

\begin{thm}
\label{thm:bb84}
For the distributed generation of BB84 states described above,
\begin{itemize}
    \item The optimal $\QCC$ for perfectly correct QSMP is $0$.
    \item The optimal $\QCC$ for $(\frac1{10},\frac1{10})$-QPSM is at least $1$.
\end{itemize}

\end{thm}

\begin{proof}
Without privacy consideration, the task can be achieved with zero quantum communication. Alice and Bob send their classical inputs $x$ and $\theta$ to Charlie, who locally prepares
$
    H^\theta |x\rangle
$. It is therefore perfectly
correct and achieves $\QCC = 0$.

Next, we show that the privacy requirement rules out zero quantum communication.
It suffices to consider the four basis inputs below with a trivial auxiliary
register, since the QPSM correctness and privacy conditions must hold in
particular for this subclass of inputs.
Suppose, toward contradiction, that there is a
$(1/10,1/10)$-QPSM protocol with zero $\QCC$. In this case, the referee receives a classical transcript $t$, which follows some distribution $D_{x,\theta}$ when Alice and Bob receive inputs $x$ and $\theta$ respectively. Let $\sigma_t$ be the qubit output by the referee on transcript $t$.

Correctness error $1/10$ implies that for every $x,\theta \in \zo$,
\[
    \frac12 \norm{\E_{t \gets D_{x,\theta}}  [\sigma_t] - H^\theta \ketbra{x} H^\theta}_1\leq \frac1{10}.
\]
Privacy error $1/10$ implies that there is a simulator $\Sim$, such that for every $x,\theta \in \zo$, 
\[
    \frac12 \norm{\E_{t \gets \Sim(H^\theta \ketbra{x} H^\theta)}  [\ketbra{t}] - \E_{t \gets D_{x,\theta}} [\ketbra{t}]}_1\leq \frac1{10}.
\]
Therefore, for every $x,\theta \in \zo$,
\[
    \frac12\norm{\E_{t \gets \Sim(H^\theta \ketbra{x} H^\theta)} [\sigma_t] - H^\theta \ketbra{x} H^\theta}_1\leq \frac1{10}+\frac1{10} = \frac1{5}.
\]
By \cref{lemma:fidelity-trace-distance}, for every $x,\theta \in \zo$,
\[
    \E_{t \gets \Sim(H^\theta \ketbra{x} H^\theta)}  
    \bra{x}H^\theta \sigma_t  H^\theta \ket{x} \geq \frac4{5}.
\]
Since the simulator $\Sim$ maps a qubit to a classical string, it can be written as a positive operator-valued measure (POVM) $\set{M_t}$, with $M_t \succeq 0$ for every $t$ and $\sum_t M_t = I$. Then on input $H^\theta \ketbra{x} H^\theta$, the simulator $\Sim$ outputs $t$ with probability $\tr\left(M_t  H^\theta \ketbra{x} H^\theta\right)$. Hence, by averaging, we have
\begin{align*}
    \frac45 \le& 
    \E_{x,\theta \gets \zo}\; \E_{t \gets \Sim(H^\theta \ketbra{x} H^\theta)}  
    \bra{x}H^\theta \sigma_t  H^\theta \ket{x}
    \\ 
    = &\E_{x,\theta \gets \zo}\; \sum_{t} \tr\left(M_t  H^\theta \ketbra{x} H^\theta\right) 
    \bra{x}H^\theta \sigma_t  H^\theta \ket{x}\\
    =& \sum_{t} \tr\left(M_t \E_{x,\theta \gets \zo} \left(H^\theta \ketbra{x} H^\theta\right) \sigma_t \left(H^\theta \ketbra{x} H^\theta\right)\right)\\
    \le & \sum_{t} \tr\left(M_t \frac38 I \right) = \frac38 \tr(I) = \frac34,
\end{align*}
which is a contradiction, where the third line is by linearity and the last inequality follows from the claim below. Therefore, we have shown that $(1/10,1/10)$-QPSM protocols for BB84 state generation must have non-zero quantum communication.
\end{proof}

\begin{claim}
\label{claim:BB84-property}
    For every qubit $\sigma$,
    \[\E_{x,\theta\gets\zo} \left(H^\theta \ketbra{x} H^\theta\right) \sigma \left(H^\theta \ketbra{x} H^\theta\right) \preceq \frac38 I. \]
\end{claim}

\begin{proof}
    Since $\sigma$ is positive semi-definite with trace $1$, we can write $\sigma=
\begin{pmatrix}
\alpha & \beta+i\gamma\\
\beta-i\gamma & 1-\alpha
\end{pmatrix}$
for some $\alpha \in [0,1]$ and $\beta,\gamma\in\bbR$ such that $\det(\sigma) = \alpha (1-\alpha) - (\beta^2 + \gamma^2) \ge 0$. Then
\begin{align*}
&\mathbb E_{x,\theta\gets{0,1}}
\left(H^\theta\ketbra{x}H^\theta\right)
\sigma
\left(H^\theta\ketbra{x}H^\theta\right)\\
={}&
\frac14\Bigl(
\ketbra{0}\sigma\ketbra{0}
+\ketbra{1}\sigma\ketbra{1}
+\ketbra{+}\sigma\ketbra{+}
+\ketbra{-}\sigma\ketbra{-}
\Bigr)\\
={}&
\frac14\left(
\begin{pmatrix}
\alpha&0\\
0&0
\end{pmatrix}
+
\begin{pmatrix}
0&0\\
0&1-\alpha
\end{pmatrix}
+
\left(\frac12+\beta\right)
\frac12
\begin{pmatrix}
1&1\\
1&1
\end{pmatrix}
+
\left(\frac12-\beta\right)
\frac12
\begin{pmatrix}
1&-1\\
-1&1
\end{pmatrix}
\right)\\
={}&
\begin{pmatrix}
\frac18+\frac{\alpha}{4}&\frac{\beta}{4}\\
\frac{\beta}{4}&\frac38-\frac{\alpha}{4}
\end{pmatrix} = \frac38 I- \frac14
\begin{pmatrix}
1-\alpha&-\beta\\
-\beta&\alpha
\end{pmatrix}
\preceq
\frac38 I,
\end{align*}
where the last relation follows because $\begin{pmatrix}
1-\alpha&-\beta\\
-\beta&\alpha
\end{pmatrix}$ is Hermitian with trace $1$ and has determinant $\alpha (1-\alpha) - \beta^2 \ge 0$, and hence positive semi-definite.
\end{proof}

\section{Quantum communication complexity with bounded entanglement}
We now turn to QPSM protocols in which the amount of pre-shared entanglement is
restricted to be linear in the input length. In the previous section, we showed
that with sufficiently large, in fact exponential, pre-shared entanglement,
every quantum channel admits a QPSM protocol whose quantum communication depends
only on the output size. We ask what remains possible under only linear
entanglement.

We first show that a class of channels, called Clifford-induced channels, admit QPSM
protocols with only one qubit of quantum communication and $O(n)$ pre-shared
entanglement.

\begin{dfn}[Clifford-induced channels]
\label{def:Clifford-induced}
    A quantum channel $\Phi$ is called Clifford-induced if there exists a unitary $U_{\Phi}$ in the Clifford group satisfying equation (\ref{eq:purification-of-channel}):
    \begin{align*}
        \Phi(\rho) = \tr_B\left[U_\Phi \left(\rho \otimes \ketbra{0}^{\ot 2m}\right) U_\Phi^\dag\right],
    \end{align*}
    where $(\reg{A}, \reg{B})$ is a partition of the output register of $U_\Phi$.
\end{dfn}

We then show that, even for arbitrary quantum channels, linear
entanglement suffices to obtain a general-purpose protocol with $O(n)$ quantum
communication. Finally, we prove that there is a quantum channel such that, when the amount of pre-shared entanglement is restricted
to be a sufficiently small linear fraction of the input length, every QSMP protocol requires $\Omega(n)$ quantum communication, even
without the privacy requirement.
These results can be summarized as the following theorem.

\begin{thm}[Bounded-entanglement QPSM]
\label{thm:ublb}
Let $n$ denote the total number of input qubits. The followings hold.
\begin{enumerate}
    \item \textbf{Clifford-induced upper bound.}
    Every Clifford-induced channel $\Phi: \cD((\bbC^2)^{\ot n}) \to \cD(\bbC^2)$ admits a $k$-party QPSM
    protocol with quantum communication complexity $1$ and entanglement complexity
    $O(n)$.
    \item \textbf{General upper bound.}
    Every quantum channel $\Phi: \cD((\bbC^2)^{\ot n}) \to \cD(\bbC^2)$ admits a $k$-party QPSM protocol with quantum communication complexity
    $O(n)$ and entanglement complexity $O(n)$.

    \item \textbf{Worst-case lower bound.}
    Any two-party QSMP protocol $\boldsymbol{\Pi}$ for the swap-test channel
    $\Phi_{\SWAP}$ with correctness error at most $1/8$ and $\mathsf{EC}(\boldsymbol{\Pi})\leq e$ satisfies
    \[
        \QCC(\boldsymbol{\Pi}) \ge n - \mathsf{EC}(\boldsymbol{\Pi}) - O(1).
    \]
    In particular, for every constant $0<c<1$, any QSMP for $\Phi_\SWAP$ with
    entanglement complexity at most $cn$ requires $\Omega(n)$ quantum communication.
    The same lower bound therefore also holds for QPSM protocols.
\end{enumerate}
\end{thm}

\subsection{A 1-qubit upper bound for Clifford-induced channels}

Let party $i$ hold an $n_i$-qubit input on Hilbert space $\cH_i$, and we write
$
    n=\sum_{i=1}^k n_i
$
for the total input length. Given a Clifford-induced quantum channel $\Phi$, let $U_\Phi$ be a unitary in the Clifford group satisfying \cref{def:Clifford-induced} that acts on both the input and an ancilla register $\cH_\anc$ initialized to the all $\ket{0}$ state. We provide a QPSM protocol for $\Phi$ using linear amount of entanglement with quantum communication complexity $m_q$. In particular, under linear amount of entanglement, the protocol uses $1$ qubit of QCC for $1$-qubit outputs and $0$ qubit of QCC for classical outputs.

\begin{thm}
\label{thm:clifford}
    For every quantum channel $\Phi: \Den(\cH_1\ot \cdots \ot \cH_k) \to \Den(\bbC^2)$ that is Clifford-induced with a unitary $U_\Phi$ in the Clifford group satisfying \cref{def:Clifford-induced}, there exists a $k$-party QPSM protocol $\boldsymbol{\Pi}$ with quantum communication complexity $\QCC(\boldsymbol{\Pi}) = 1$ and entanglement complexity $\mathsf{EC}(\boldsymbol{\Pi})=O(n)$.
\end{thm}

\begin{protocol}{}
\label{protocol:QPSM:1-QCC:kn-entanglement}
Every party $i>1$ pre-share $n_i$ EPR pairs $(e^{i \to 1}_i, e^{i \to 1}_1)$ with party $1$, where party $i$ and party $1$ hold $e^{i \to 1}_i$ and $e^{i \to 1}_1$ respectively. 
\begin{itemize}
    \item Every party $i>1$ teleports its input $\rho_i$ to party $1$ using the EPR pairs pre-shared with party $1$.  Let
    $
        P_i\in\PauliGroup_{\cH_i}
    $
    be the Pauli correction party $i$ obtains by this teleportation.
    \item Party $1$ holds the joint state
    \[\rho' = (\rho_1, e^{2\to 1}_1,\dots, e^{k\to 1}_1),\]
    computes $\Phi$, and sends the result $\sigma := \Phi(\rho')$ to the referee.
    \item Every party $i>1$ computes $$P'_i := U_\Phi (I_{\cH_1}\otimes\cdots\otimes P_i^\dag \otimes\cdots\otimes I_{\cH_k} \otimes I_{\cH_\anc}) U_\Phi^\dag, $$
    which is guaranteed to be a Pauli operator since $U_\Phi$ is a Clifford, \ie a normalizer of the Pauli group. Let $(\reg{A},\reg{B})$ be the registers described in \cref{def:Clifford-induced}. Party $i$ decomposes $P'_i = \onreg{P'_{i,A}}{A} \otimes \onreg{ P'_{i,B}}{B}$ and sends $P'_{i,A}$ to the referee.
    \item The referee computes $P'_{k,A} \cdots P'_{2,A} \sigma P'^\dag_{2,A} \cdots P'^\dag_{k,A}$ and outputs the result.
\end{itemize}
\end{protocol}

The referee receives an output of $\Phi$ and the classical description of Pauli operators $P'_{2,A},\dots,P'_{k,A}$. Hence, the quantum communication complexity is $1$ qubit. The pre-shared entanglement consists of $n_i$ EPR pairs between party $1$ and party $i$ for every $i>1$, and hence the entanglement complexity is 
$O(n)$.

\begin{lemma}[Correctness]
Protocol~\ref{protocol:QPSM:1-QCC:kn-entanglement} is perfectly correct.
\end{lemma}

\begin{proof}
Fix an arbitrary auxiliary register $\reg{R}$ and an arbitrary joint input
state $\onreg{\rho}{\cH_1\cdots \cH_k R}$. Set $\onreg{\rho}{\cH_1\cdots \cH_k \cH_\anc R} = \onreg{\rho}{\cH_1\cdots \cH_k R} \otimes \onreg{\ketbra{0}^{\otimes n_\anc}}{\cH_\anc}$ if $\cH_\anc$ holds $n_\anc$ qubits.
Let
\[
    \bar P_i
    :=
    I_{\cH_1}\otimes\cdots\otimes P_i\otimes\cdots
    \otimes I_{\cH_k}\otimes I_{\cH_\anc}
\]
be the Pauli error induced by teleporting party $i$'s input to party $1$. Let
$
    \bar P:=\bar P_k\cdots \bar P_2 .
$
By the correctness of teleportation, the post-teleportation
joint state would be
\[
\onreg{\rho'\;}{\cH_1\dots\cH_k\cH_\anc R}
=
(\bar P\otimes I_R)\;
\onreg{\rho}{\cH_1\dots\cH_k\cH_\anc {R}}
\;(\bar P^\dag\otimes I_R).
\]
The quantum message jointly with the auxiliary register is
\[
\onreg{\sigma}{A{R}}
= 
\tr_B\!\left[
(U_\Phi\otimes I_R)\;\onreg{\rho'}{\cH_1\dots\cH_k\cH_\anc {R}}\;
(U_\Phi^\dag\otimes I_R)
\right].
\]

For each $i>1$, the announced correction is
$
    P'_i=U_\Phi\bar P_i^\dag U_\Phi^\dag
$.
Since $U_\Phi$ normalizes the Pauli group, $P'_i$ is a Pauli operator, up to a
global phase, and we may write
$
    P'_i=P'_{i,A}\otimes P'_{i,B}.
$
Let
$
    P'_A:=P'_{k,A}\cdots P'_{2,A}
$
and
$
    P'_B:=P'_{k,B}\cdots P'_{2,B}.
$
Then, up to a global phase,
\[
    P'_A\otimes P'_B
    =
    P'_k\cdots P'_2
    =
    U_\Phi\bar P^\dag U_\Phi^\dag .
\]
The joint state of the referee's output and the auxiliary register is
\[
\begin{aligned}
    &(P'_A \ot I_R) \onreg{\sigma}{A{R}} (P'^\dag_A \ot I_R)\\
    =&
    (P'_A \ot I_R)
    \tr_B\!\left(
(U_\Phi\otimes I_R)\;\onreg{\rho'}{\cH_1\dots\cH_k\cH_\anc {R}}\;
(U_\Phi^\dag\otimes I_R)
\right)
    (P'^\dag_A \ot I_R)                                      \\
    =&
    \tr_{B}\!\left(
        \left((P'_A\otimes I_B)
        U_\Phi\bar P \otimes I_R\right)(\onreg{\rho}{\cH_1\dots\cH_k\cH_\anc {R}})
        \left(\bar P^\dag U_\Phi^\dag
        (P'^\dag_A\otimes I_B) \otimes I_R\right)
    \right)                                             \\
    =&
    \tr_{B}\!\left(
        \left((P'_A\otimes P'_B)
        U_\Phi\bar P \otimes I_R\right)(\onreg{\rho}{\cH_1\dots\cH_k\cH_\anc {R}})
        \left(\bar P^\dag U_\Phi^\dag
        (P'^\dag_A\otimes P'^\dag_B) \otimes I_R\right)
    \right)                                             \\
    =&
    \tr_{B}\!\left(
        \left(U_\Phi\bar P^\dag U_\Phi^\dag
        U_\Phi\bar P \otimes I_R\right)(\onreg{\rho}{\cH_1\dots\cH_k\cH_\anc {R}}) \left(
        \bar P^\dag U_\Phi^\dag
        U_\Phi\bar P U_\Phi^\dag \otimes I_R\right)
    \right)                                             \\
    =&
    \tr_{B}\!\left(
        \left(U_\Phi\otimes I_R\right)(\onreg{\rho}{\cH_1\dots\cH_k\cH_\anc {R}})\left(U_\Phi^\dag\otimes I_R\right)
    \right)                                             \\
    =&
    \left(\Phi \otimes \cI_R\right)(\rho).
\end{aligned}
\]
Hence the protocol is perfectly correct.

\end{proof}
\begin{lem}[Privacy]
\cref{protocol:QPSM:1-QCC:kn-entanglement} is perfectly private. \end{lem}

\begin{proof}
Fix an arbitrary auxiliary register $\reg{R}$ and an arbitrary joint input
state $\onreg{\rho}{\cH_1\cdots \cH_k {R}}$.
Write
$
    P'_A:=P'_{k,A}\cdots P'_{2,A}.
$
The referee's view consists of the quantum message $\onreg{\sigma}{A}$ together with the
classical Paulis $P'_{2,A},\dots,P'_{k,A}$. By the correction identity from the correctness proof,
$(P'_A \ot I_R) \onreg{\sigma}{A{R}} (P'^\dag_A \ot I_R) = (\Phi \ot \cI_R)(\rho)$, or equivalently,
$\onreg{\sigma}{A{R}} = (P'^\dag_A \ot I_R) (\Phi \ot \cI_R)(\rho) (P'_A\ot I_R)$.
Hence the referee's view together with the auxiliary register can be written as
\[
    (\View \ot \cI_R)(\rho)
    =
    \left(
        (P'^\dag_A \ot I_R)(\Phi\ot \cI_R)(\rho) (P'_A \ot I_R),
        P'_{2,A},\dots,P'_{k,A}
    \right).
\]

The teleportation outcomes $P_2,\dots,P_k$ are uniformly random and independent
of the input state $\rho$. Since each $P'_{i,A}$ is a deterministic function of
$P_i$ and the public Clifford $U_\Phi$, the distribution of the tuple
$
    (P'_{2,A},\dots,P'_{k,A})
$
is also input-independent.

Therefore, given only the output register $\reg{A}$ of the state $\sigma_{\out}:=\onreg{(\Phi\ot\cI_R)(\rho)}{A{R}}$, a simulator can
sample fresh teleportation Paulis $\widehat P_2,\dots,\widehat P_k$ from the
uniform distribution, compute
\[
    \widehat P'_i
    :=
    U_\Phi
    (I_{\cH_1}\otimes\cdots\otimes \widehat P_i^\dag
    \otimes\cdots\otimes I_{\cH_k}\otimes I_{\cH_\anc})
    U_\Phi^\dag,
\]
decompose $\widehat P'_{i} =
    \widehat P'_{i,A}\otimes \widehat P'_{i,B}$,
apply
$
    (\widehat P'_A)^\dag:=(\widehat P'_{k,A}\cdots \widehat P'_{2,A})^\dag
$ to register $\reg{A}$, and output $\reg{A}$ along with $(\widehat P'_{2,A},\dots,\widehat P'_{k,A})$. The joint output of the simulator and the auxiliary register is
\begin{align*}
    (\Sim \ot \cI_R)(\sigma_{\out}) =&\left(
        (\widehat P'_A \ot I_R)^\dag \sigma_{\out} (\widehat P'_A\ot I_R),
        \widehat P'_{2,A},\dots,\widehat P'_{k,A}
    \right)\\
    =&\left(
        (\widehat P'_A \ot I_R)^\dag {(\Phi\ot\cI_R)(\rho)} (\widehat P'_A\ot I_R),
        \widehat P'_{2,A},\dots,\widehat P'_{k,A}
    \right)
    =(\View \ot \cI_R)(\rho).
\end{align*}
Thus the protocol is perfectly private.
\end{proof}

\subsection{A linear upper bound for general quantum channels}
\label{sec:bounded-upper}

\begin{theorem}
\label{thm:linear-qpsm-upper}
    For every quantum channel
    $
        \Phi:\Den(\cH_1\ot\cdots\ot\cH_k)\to \Den(\bbC^2)
    $
    that outputs one qubit, there is a $k$-party QPSM protocol $\boldsymbol{\Pi}$ with quantum communication complexity
    $
        \QCC(\boldsymbol{\Pi})=O(n)
    $
    and entanglement complexity
    $
        \mathsf{EC}(\boldsymbol{\Pi})=O(n).
    $
\end{theorem}

We achieve a general linear upper bound by hiding quantum information with random unitaries, the technique of which is conceptually related to the group randomizing technique of \cite{brakerski2022quantum}.
However, our setting has an additional challenge: different parties hold
different teleportation corrections, and the messages they send must still
be jointly simulatable from the final output alone. We overcome this by having each party absorb its
own Pauli correction into the classical description of a random unitary.
These descriptions can then be composed by the referee to recover the
desired output, while the randomization ensures that the joint transcript can be simulated from the output alone.

\begin{protocol}{}
\label{protocol:QPSM:kn-QCC:kn-entanglement}
Every party $i>1$ pre-share $n_i$ EPR pairs $(e^{i \to 1}_i, e^{i \to 1}_1)$ with party $1$, where party $i$ and party $1$ hold $e^{i \to 1}_i$ and $e^{i \to 1}_1$ respectively. The parties also pre-share independent Haar random unitaries
$
    E_1,\ldots,E_{k-1}
$
on the Hilbert space $\cH_1\otimes\cdots\otimes\cH_k \otimes \cH_{\anc}$.
\begin{itemize}
    \item Every party $i>1$ teleports its input $\rho_i$ to party $1$ using the EPR pairs pre-shared with party $1$.  Let
    $
        P_i\in\PauliGroup_{\cH_i}
    $
    be the Pauli correction party $i$ obtains by this teleportation.
    \item Party $1$ holds the joint state
    \[\rho' = (\rho_1, e^{2\to 1}_1,\dots, e^{k\to 1}_1, \ketbra{0}^{\ot n_\anc}),\]
    applies $E_1$, and sends the resulting quantum state
    \[
        \sigma
        =
        E_1 \rho' E_1^\dag
    \]
    to the referee.
    \item Every party $1<i<k$ sends the classical description\footnote{One can consider either a model where real numbers can be sent through classical communication channels without precision errors, or a protocol where the parties send real numbers approximately. The latter would additionally incur arbitrarily small correctness and privacy errors.} of the unitary
    \[
        U_i := E_i
        (I_{\cH_1}\otimes\cdots\otimes P_i^\dag \otimes\cdots\otimes I_{\cH_k}
        \otimes I_{\cH_\anc})
        E_{i-1}^\dag
    \]
    to the referee.
    \item Party $k$ samples a random Pauli $P'$ and sends the classical description of the unitary
    \[
        U_k
        :=
        (\onreg{I_{m}}{A} \otimes \onreg{P'}{B}) \onreg{U_\Phi}{A,B}
        (I_{\cH_1}\otimes\cdots\otimes I_{\cH_{k-1}} \otimes P_k^\dag
        \otimes I_{\cH_\anc})
            E_{k-1}^\dag 
    \]
    to the referee.
    \item The referee outputs
    \[\tr_{B}[U_k \cdots U_2 \sigma U_2^\dag \cdots U_k^\dag]\]
\end{itemize}
\end{protocol}

The only quantum message sent to the referee is the state $\sigma$ sent from party $1$, and therefore
$
    \QCC(\boldsymbol{\Pi})=O(n).
$
The pre-shared entanglement consists of $n_i$ EPR pairs between party $1$ and party $i$ for every $i>1$, and hence
$
    \mathsf{EC}(\boldsymbol{\Pi})= O(n).
$

\begin{lem}[Correctness]
\cref{protocol:QPSM:kn-QCC:kn-entanglement} is perfectly correct.
\end{lem}

\begin{proof}
    Fix an arbitrary auxiliary register $\reg{R}$ and an arbitrary joint input
    state $\onreg{\rho}{\cH_1\cdots \cH_k {R}}$. Set $\onreg{\rho}{\cH_1\cdots \cH_k \cH_\anc {R}} = \onreg{\rho}{\cH_1\cdots \cH_k {R}} \otimes \onreg{\ketbra{0}^{\otimes n_\anc}}{\cH_\anc}$ if $\cH_\anc$ holds $n_\anc$ qubits. Set \[\bar P:=I_{\cH_1}\otimes P_2\otimes\cdots\otimes P_k\otimes I_{\cH_{\anc}}.\]
    By the correctness of teleportation, the joint message and auxiliary state is\[\onreg{\sigma}{\cH_1,\dots,\cH_k,\cH_\anc,{R}} = (E_1\otimes I_R)(\bar P\otimes I_R)\onreg{\rho}{\cH_1\cdots \cH_k \cH_\anc {R}} (\bar P^\dag\otimes I_R)(E_1^\dag\otimes I_R).\]
    For $2\leq i<k$, define\[\bar P_{>i}:=I_{\cH_1}\otimes\cdots\otimes I_{\cH_i}\otimes P_{i+1}\otimes\cdots\otimes P_k\otimes I_{\cH_{\anc}}.\]
    It is direct to see, by induction on $i \ge 2$, that we have
    \begin{align*}
    &(U_i \cdots U_2 \ot I_R)\;\onreg{\sigma}{\cH_1,\dots,\cH_k,\cH_\anc,{R}}\;(U_2^\dag \cdots U_i^\dag \ot I_R)\\
    =& (E_i\otimes I_R)(\bar P_{>i}\otimes I_R)\;\onreg{\rho}{\cH_1\cdots \cH_k \cH_\anc {R}}\; (\bar P_{>i}^\dag\otimes I_R)(E_i^\dag\otimes I_R),
    \end{align*}
    and for $i=k$, the state $(U_i \cdots U_2 \ot I_R)\;\onreg{\sigma}{\cH_1,\dots,\cH_k,\cH_\anc,{R}}\;(U_2^\dag \cdots U_i^\dag \ot I_R)$ is equal to
    \[ \bigl((I_A\otimes P'_B)U_\Phi\otimes I_R\bigr)\;\onreg{\rho}{\cH_1\cdots \cH_k \cH_\anc {R}}\; \bigl((I_A\otimes P'_B)U_\Phi\otimes I_R\bigr)^\dag. \]
    After the referee traces out 
    register $\reg{B}$, the joint output state of the referee and the auxiliary register is
    \begin{align*}
        \tr_B\left[\bigl((I_A\otimes P'_B)U_\Phi\otimes I_R\bigr)\;\onreg{\rho}{\cH_1\cdots \cH_k \cH_\anc {R}}\; \bigl((I_A\otimes P'_B)U_\Phi\otimes I_R\bigr)^\dag\right]
        = (\Phi \ot \cI_R)(\onreg{\rho}{\cH_1\dots\cH_k {R}}).
    \end{align*}
    This proves perfect correctness.
\end{proof}

\begin{lem}[Privacy]
\cref{protocol:QPSM:kn-QCC:kn-entanglement} is perfectly private.
\end{lem}

\begin{proof}
    Fix an arbitrary auxiliary register $\reg{R}$ and an arbitrary joint input
    state $\onreg{\rho}{\cH_1\cdots \cH_k {R}}$. Set $\onreg{\rho}{\cH_1\cdots \cH_k \cH_\anc {R}} = \onreg{\rho}{\cH_1\cdots \cH_k {R}} \otimes \onreg{\ketbra{0}^{\otimes n_\anc}}{\cH_\anc}$ if $\cH_\anc$ holds $n_\anc$ qubits.
    For every $P'$, since $E_1,\dots,E_{k-1}$ are independent Haar random unitaries,
    the distribution of $U_{k}$ is Haar random by its definition and translation invariance of the Haar measure, and by induction, for all $i>1$, the distribution of $U_i$ conditioned on $U_{i+1},\dots,U_k$ is also Haar random from the definition of $U_i$.
    Therefore, $(U_2,\dots,U_k)$ are independent Haar random unitaries and does not depend on $P'$. Thus, $P'$ remains uniform even conditioning on the classical transcript $(U_2,\dots,U_k)$. 
    Set \[\bar P:=I_{\cH_1}\otimes P_2\otimes\cdots\otimes P_k\otimes I_{\cH_{\anc}}.\] 
    Since $U_k \cdots U_2 = (I \otimes P') U_\Phi \bar{P}^\dag E_1^\dag$, we can rewrite $E_1$ as
    \[E_1 = (U_2^\dag \cdots U_k^\dag) (I \otimes P') U_\Phi \bar{P}^\dag.\]
    By the correctness of teleportation, the joint quantum message and auxiliary state is\begin{align*}
        \onreg{\sigma}{\cH_1,\dots,\cH_k,\cH_\anc,{R}} &= (E_1\otimes I_R)(\bar P\otimes I_R)\;\onreg{\rho}{\cH_1\cdots \cH_k \cH_\anc {R}}\; (\bar P^\dag\otimes I_R)(E_1^\dag\otimes I_R)\\
        &=\left((U_2^\dag \cdots U_k^\dag) (I \otimes P') U_\Phi \ot I_R\right) \;\onreg{\rho}{\cH_1\cdots \cH_k \cH_\anc {R}}\; \left((U_2^\dag \cdots U_k^\dag) (I \otimes P') U_\Phi \ot I_R\right)^\dag.
    \end{align*}
    The referee receives the state $\onreg{\sigma}{\cH_1,\dots,\cH_k,\cH_\anc}$ and the unitaries $U_2,\dots,U_k$ but never receives $P'$.
    By Pauli mixing (\cref{lemma:Pauli-mixing}), the joint state of the referee's view and the auxiliary state is
    {\small
    \begin{align*}
        &\E_{P'} \left[\left(\left((U_2^\dag \cdots U_k^\dag) (I \otimes P') U_\Phi \ot I_R\right) \onreg{\rho}{\cH_1\cdots \cH_k \cH_\anc {R}} \left((U_2^\dag \cdots U_k^\dag) (I \otimes P') U_\Phi \ot I_R\right)^\dag, U_2,\dots,U_k\right)\right]\\
        =& \left((U_2^\dag \cdots U_k^\dag \ot I_R) \left(\tr_B [({U_\Phi} \ot I_R)\;\rho\; (U_\Phi^{\dag} \ot I_R)] \otimes \psi_{\mathsf{maxmixed}}\right) (U_k \cdots U_2 \ot I_R), U_2, \dots, U_k\right)\\
    =& \left((U_2^\dag \cdots U_k^\dag \ot I_R) \left((\Phi\ot \cI_R)(\rho)\otimes \psi_{\mathsf{maxmixed}}\right) (U_k \cdots U_2 \ot I_R), U_2, \dots, U_k\right).
    \end{align*}
    }Hence, we can define a simulator $\Sim$ that, given the register $\reg{A}$ of the state $\onreg{(\Phi\ot \cI_R)(\rho)}{A{R}}$, appends $\onreg{\psi_{\mathsf{maxmixed}}}{B}$, samples Haar random unitaries $U_2,\dots,U_k$, applies $(U_k\cdots U_2)^\dag$ on $\reg{AB}$, and outputs the resulting quantum state together with the classical descriptions of the sampled unitaries. This produces exactly the state above, and we have
    \[\bigl(\Sim\otimes\cI_R\bigr)\bigl(\Phi\otimes\cI_R\bigr)(\rho) = (\View\otimes\cI_R)(\rho).\]
    Therefore, the protocol satisfies perfect privacy.
\end{proof}

\subsection{A linear lower bound in the worst case}
We establish a lower bound on the quantum communication complexity when the pre-shared entanglement complexity is small. It applies to any QPSM protocol computing the swap-test channel $\Phi_\SWAP(\psi,\phi) := \frac{1+\tr(\psi \phi)}{2} \ketbra{0} + \frac{1-\tr(\psi \phi)}{2} \ketbra{1}$ on $\frac{n}{2}$-qubit states $\psi$ and $\phi$.

\begin{thm}
\label{thm:linear-qpsm-lower}
    For any two-party QSMP protocol $\boldsymbol{\Pi}$ for $\Phi_\SWAP$ with $1/8$-correctness-error,
    \[\QCC(\boldsymbol{\Pi}) \ge n - \mathsf{EC}(\boldsymbol{\Pi}) - O(1).\]
\end{thm}

\begin{proof}
Set $n' = \frac{n}{2}$.
Let $\boldsymbol{\Pi}$ be such a QSMP protocol, and let Alice and Bob send $q_A$ and $q_B$ qubits, respectively. Thus
\[
    \QCC(\boldsymbol{\Pi})=q_A+q_B.
\]
Write $e_A=|S_A^{\mathrm q}|$ and $e_B=|S_B^{\mathrm q}|$ for the sizes of their pre-shared quantum registers, so that $\mathsf{EC}(\boldsymbol{\Pi}) = e_A+e_B$.
We use the QSMP to construct a protocol for solving the decision DIPE$_{1,n'}$ problem as follows, where we make Bob simulate the referee. Alice receives her input, locally prepares the shared entangled state, and sends Bob his share. This costs $e_B$ qubits of quantum communication; the classical part of the pre-shared state can be sent for free. Alice then computes and sends her QSMP message $m_A$ to Bob, which costs another $q_A$ qubits of quantum communication. Bob uses his input and the pre-shared state he receives from Alice to compute his QSMP message $m_B$. Bob then applies the referee's final computation on $(m_A, m_B)$ to obtain the final output. Since the QSMP computes $\Phi_\SWAP$ with $1/8$-correctness-error, and $\Phi_\SWAP$ solves the decision DIPE$_{1,n'}$ with correctness $1$ and soundness error exponentially close to $1/2$, the two-party protocol constructed above solves the decision distributed inner product estimation problem with correctness $7/8$ and soundness error $6/8$, with quantum communication $q_A+e_B$. By \cref{thm:DIPE-lower-bound}, we have $1=c=\Omega(\sqrt{2^{n'-(q_A+e_B)}})$, and therefore $q_A + e_B \ge n' - O(1)$.
By symmetry, we can make Alice simulate the referee instead, and deduce that $q_B+e_A\geq n'-O(1)$. Adding the two inequalities gives
\[
    \QCC(\boldsymbol{\Pi})=q_A+q_B
    \geq 2n'-(e_A+e_B)-O(1)
    = n-\mathsf{EC}(\boldsymbol{\Pi})-O(1).
\]
\end{proof}

Our upper and lower bounds on QPSM and our equivalence between QDRE and QPSM (\cref{thm:QDRE-QPSM-equivalence}) yield the following corollaries.

\begin{cor}
The following bounds hold for QDREs.
\begin{enumerate}
\item \textbf{Large setup states.}
Every quantum channel $\Phi:\Den((\bbC^2)^{\ot n})
\to
\Den((\bbC^2)^{\otimes m_q}; \{0,1\}^{m_c})$ admits a QDRE whose quantum encoding size is at most $m_q$ with arbitrarily small error. In particular, every classical-output quantum channel admits a QDRE with zero quantum encoding size.
\item \textbf{Bounded setup states.} Every quantum channel
$\Phi:\Den((\bbC^2)^{\ot n})\to\Den(\bbC^2)$ on $n$ input qubits admits a QDRE with $O(n)$ setup-state size and $O(n)$ quantum encoding size. Moreover, there exists a quantum channel
$\Phi:\Den((\bbC^2)^{\otimes n})
    \to
    \Den(\bbC^2)$
and constants $c,c'>0$ such that any QDRE for $\Phi$ with setup-state size at most $c\cdot n$ and correctness error within $c'$ must have quantum encoding size $\Omega(n)$.

\item \textbf{Clifford-induced quantum channels.} Every Clifford-induced quantum channel
$\Phi:\Den((\bbC^2)^{\ot n})\to\Den(\bbC^2)$ admits a QDRE with $O(n)$ setup-state size and $1$-qubit quantum encoding size.
\end{enumerate}
\end{cor}

\section*{\ifdefined\ShowAuthor  Acknowledgment and \fi
AI Disclosure}
\ifdefined\ShowAuthor
The authors would like to thank Noam Mezor, Shang-Hua Teng, Andrea Coladangelo, Lijie Chen, and Umesh Vazirani for useful discussions. \fi
The authors used ChatGPT 5.5 to help draw \cref{fig:qpsm-pbt-overview}, ChatGPT 5.6 to identify a bug in the proof of \cref{thm:bb84} in a preliminary version, and ChatGPT 6 to assist with proofreading, minor editorial revisions, and identifying relevant literature. 
\ifdefined\ShowAuthor
Part of the work was done when Miryam was a student at USC. Part of the work was carried out when Miryam Huang and Er-Cheng Tang visited the Simons Institute for the Theory of Computing in 2025 and 2026.
\fi

\ifdefined\LLNCS
\bibliographystyle{splncs04}
\else
\bibliographystyle{alpha}
\fi
\bibliography{references}

\ifdefined\FullVersion
\appendix
\fi

\appendix
\end{document}